\documentclass[journal,peerreview,12pt]{IEEEtran}
\usepackage{amsmath,amsfonts,amssymb,amsthm}
\usepackage{xurl}
\usepackage{graphicx}
\usepackage{cite}
\usepackage{bbm}
\usepackage{physics}
\usepackage{mleftright}
\usepackage{mathtools}
\usepackage{algpseudocode}
\usepackage{algorithm}
\usepackage{hyperref}
\usepackage{placeins}

\newtheorem{theorem}{Theorem}

\newtheorem{proposition}[theorem]{Proposition}

\newtheorem{remark}{Remark}
\newtheorem*{theorem*}{Theorem}
\newtheorem*{corollary*}{Corollary}
\newtheorem*{proposition*}{Proposition}
\newtheorem*{remark*}{Remark}

\theoremstyle{definition}
\newtheorem{example}{Example}
\newtheorem{definition}{Definition}
\newtheorem*{definition*}{Definition}

\newcommand*\BitwiseAND{\mathbin{\&}}
\newcommand*\LogicalAND{\mathbin{\&\&}}
\newcommand*\LogicalOR{\mathbin{||}}
\newcommand*\BitwiseOR{\mathbin{|}}
\newcommand*\ShiftLeft{\ll}
\newcommand*\ShiftRight{\gg}
\newcommand*\LogicalNOT{!}

\makeatletter
\NewCommandCopy\@@pmod\pmod
\DeclareRobustCommand{\pmod}{\@ifstar\@pmods\@@pmod}
\def\@pmods#1{\mkern4mu({\operator@font mod}\mkern 6mu#1)}
\makeatother

\allowdisplaybreaks
\begin{document}
	
\title{Higher-Order Staircase Codes}

\author{Mohannad~Shehadeh,~\IEEEmembership{Graduate Student Member,~IEEE},
	Frank~R.~Kschischang,~\IEEEmembership{Fellow,~IEEE},
	Alvin~Y.~Sukmadji,~\IEEEmembership{Graduate Student Member,~IEEE},
	and William~Kingsford
	\thanks{Submitted to \emph{IEEE Transactions on Information Theory} on May.~30, 2024.
			Parts of this work \cite{mohannad-ofc} were presented at
			the 2024 Optical Fiber Communication Conference.}
	\thanks{Mohannad~Shehadeh, Frank~R.~Kschischang, and Alvin~Y.~Sukmadji are with the Edward S. Rogers Sr. Department of Electrical \& Computer Engineering, University of Toronto, Toronto, ON M5S 3G4, Canada (emails: \{mshehadeh, frank, asukmadji\}@ece.utoronto.ca).}
	\thanks{William~Kingsford is an unaffiliated scholar from Palmerston North, New Zealand (email: will.kingsford@gmail.com).}
}

\maketitle

\begin{abstract}
	We generalize
	staircase codes 
	and tiled diagonal zipper codes,
	preserving their key properties while allowing each
	coded symbol to be protected by arbitrarily 
	many component codewords 
	rather than only two. 
	This generalization which we term ``higher-order staircase codes'' 
	arises from the marriage 
	of two distinct combinatorial objects:
	difference triangle sets and finite-geometric nets, 
	which have typically been applied separately 
	to code design. We demonstrate 
	one possible realization of these codes, obtaining 
	powerful, high-rate, low-error-floor, and low-complexity
	coding schemes based on simple iterative syndrome-domain
	decoding of coupled Hamming component codes. 
	We anticipate that the proposed codes
	could improve performance--complexity--latency tradeoffs 
	in high-throughput communications
	applications, most notably fiber-optic, 
	in which classical staircase codes and 
	zipper codes have been applied. We
        consider the construction of
	difference triangle sets having minimum scope and sum-of-lengths, which lead
	to memory-optimal realizations of
	higher-order staircase codes. 
	These results also enable memory reductions for 
	early families of convolutional codes 
	constructed from difference triangle sets.
\end{abstract}


\section{Overview}
\IEEEPARstart{S}{taircase} codes \cite{smith} 
and related product-like codes \cite{sukmadji}
have provided a highly-competitive paradigm
for energy-efficient, 
high-rate, high-throughput, and 
high-reliability, error control coding.
Such codes use syndrome-domain iterative
algebraic or 
bounded distance decoding of 
algebraic component
codes to enable high decoder throughputs,
while maintaining relatively low internal decoder data-flow. 
Furthermore, their 
block-convolutional 
structures
translate to natural 
pipelining and parallelism, the combination 
of which is essential for efficient hardware designs.
Moreover, their typically 
combinatorial constructions 
ensure control over error floors which 
must usually be below $10^{-15}$ 
in the relevant applications. 
Most notably, such codes are 
used in fiber-optic communications \cite[Ch.~7]{optical-networks-handbook},
but are also relevant to flash storage 
applications \cite{NAND-staircase} 
which increasingly demand 
similar code characteristics.

In this paper, we provide a
fruitful generalization of staircase codes 
which arises
from the combination of 
\emph{difference triangle sets} (DTSs) \cite[Part~VI:~Ch.~19]{Colbourn} 
and finite-geometric \emph{nets} 
\cite[Part~III:~Ch.~3]{Colbourn}.
Both DTSs and nets are well-studied 
combinatorial objects
and both have been separately
applied to code design, e.g., in \cite{CSOC,NB-DTS} and 
\cite{SJ-PG,TD-LDPC} respectively. In these
examples, DTSs lead naturally to convolutional
codes while nets lead naturally to block codes.
It is unsurprising then that their combination
leads to a generalization of staircase
codes which themselves combine aspects
of block and convolutional coding.

We will proceed
to summarize at a high level
our key construction.
A \emph{higher-order staircase code} is determined
by three intertwined objects:  a DTS, a net,
and a \emph{component code}.  Each of these
objects is characterized by two parameters,
the meanings of which will be defined later.
The relationships between these parameters
and properties of the resulting code are
as follows:
Given an $(L,M)$-DTS, an $(M+1,S/L)$-net,
and a component code of length $(M+1)S$
and dimension $(M+1)S-r$, one obtains
a rate $1 - r/S$
spatially-coupled code 
in which every symbol
is protected by $M+1$ component
codewords and any pair of distinct 
component codewords share at most one symbol. 
Alternatively, 
the Tanner graph of the code 
is an infinite semiregular 
$4$-cycle-free bipartite graph 
with variable
node degree $M+1$ and generalized
constraint node degree $(M+1)S$.
The parameter $L$ controls the balance of
memory depth against parallelism 
while keeping all 
of the aforementioned code properties 
completely unchanged. A fourth parameter $C$
controls the parallelism obtained by circularly coupling multiple 
copies of a higher-order staircase code together. The parameters
$M$ and $C$ provide mechanisms for improving error floor 
and threshold performance respectively without altering
the component code. 

The proposed code family recovers
as particular cases some existing code families
and interpolates between them, thus providing a unified
generalization:
\begin{itemize}
	\item When $L = M = C = 1$, 
	the classical staircase codes of
	Smith et al.\ \cite{smith} are recovered.
	\item When $L \geq 1$ and $M = C = 1$, the
	tiled diagonal zipper codes 
	of Sukmadji et al.\ \cite{sukmadji-cwit} are recovered.
	\item When $L \geq 1$, $M = 1$, and $C = 2$, the OFEC code family \cite{OFEC} is recovered.
	\item When $L \geq 1$, $M = 1$, and $C \geq 2$, the multiply-chained tiled diagonal zipper codes of 
	Zhao et al.\ \cite{MC-TDZC} are recovered.
	\item When $S/L = M = C = 1$, the continuously interleaved codes of
	Coe \cite{CI-patent} are recovered.
	\item When $S/L = C = r = 1$, a
	recursive (rather than feedforward)
	version of the self-orthogonal convolutional
	codes of Robinson and Bernstein \cite{CSOC} is recovered.
\end{itemize}
When $M > 1$ and $L = 1$, we obtain
a generalization of staircase codes with 
higher variable degree which was included in a
conference version of this work \cite{mohannad-ofc}.
A similar construction is included in the recently-issued
US patent \cite{Geyer-patent}, which recognizes the role of 
$(1,M)$-DTSs in such a generalization of staircase
codes, but makes no explicit use of finite-geometric nets
or general $(L,M)$-DTSs.
Importantly, when $L > 1$ and $M > 1$, 
our construction becomes a 
higher-variable-degree generalization
of tiled diagonal zipper codes \cite{sukmadji-cwit}
which reduces encoding and decoding memory 
while also providing
flexibility in aspects of hardware implementation.

Constructing memory-optimal higher-order staircase codes
reduces to a natural extension of the classically studied 
\emph{minimum scope DTS problem} \cite[Part~VI:~Ch.~19]{Colbourn}
in which we additionally seek to minimize 
a new objective:  \emph{sum-of-lengths}.
We study this problem, exhibiting new DTSs 
that also have implications for earlier 
work on classical convolutional codes constructed from DTSs.

The proposed codes, 
like their predecessors in \cite{smith,sukmadji},
combine aspects of convolutional codes, 
low-density parity-check (LDPC) codes, and algebraic codes.
They can be described as an instance of
zipper codes \cite{sukmadji}, spatially-coupled 
generalized LDPC (SC-GLDPC) codes \cite{SC-GLDPC},
or even simply as (highly complex) convolutional codes. 
However, the proposed codes are 
very highly structured relative to what is
typically admitted under these frameworks. 
This means that the simplest possible description 
will be one which is self-contained and does 
not explicitly appeal to them.
Without discounting the potential significance of 
these connections, 
we will limit our goal in this paper to providing
such a description. 

The key advantages of the code structure considered is that it simultaneously enables a large degree of parallelism due to a block-based construction, 
high performance due to spatial coupling, and control over error floors due to the combinatorial construction. The recovery of established high-throughput coding techniques 
\cite{smith, sukmadji-cwit, OFEC, CI-patent, MC-TDZC} as special cases is 
then unsurprising since these properties are precisely what is demanded by such applications.

The remainder of this paper is
organized as follows: We begin 
in Section \ref{gsc-section}
with a special case of our construction
corresponding to $L = 1$ along with
key tools and definitions. In 
Section \ref{instantiation-section},
we illustrate one possible realization of the
construction of Section \ref{gsc-section}, discussing
various practical aspects and a high-throughput simulation
technique. Section \ref{instantiation-section} 
also proposes techniques for
efficiently decoding systematic Hamming codes
without lookup tables (LUTs). Following this,
we provide the general, complete 
construction of higher-order staircase codes 
corresponding to $L \geq 1$ in Section 
\ref{Generalized-TDZC}. We follow this
with some simulation-based demonstrations 
in Section \ref{hosc-sim-section}. 
In Section \ref{DTS-section},
we study the problem of constructing 
DTSs that correspond to memory-optimal instances
of higher-order staircase codes. A variety of such DTSs are exhibited in the Appendix.
Finally, we describe our DTS search technique in 
Section \ref{computer-search-technique} 
before concluding in Section \ref{Conclusion} 
with suggestions for future work.

\section{Generalized Staircase Codes with Arbitrary Bit Degree ($L=1$, $M\geq 1$)}\label{gsc-section}

In this section, we provide a generalization
of staircase codes \cite{smith} in which 
each coded bit (or nonbinary symbol) is
protected by arbitrarily many component codewords
rather than only two. The goal is to do this
while preserving the property of classical staircase codes,
essential to error floor performance, 
that distinct component codewords intersect in at most one bit. 
This construction
will arise from the combination of
a $(1,M)$-DTS and an $(M+1,S)$-net. 
A $(1,M)$-DTS is equivalent to a \emph{Golomb ruler}
of order $M+1$, to be defined shortly, with the
definition of a general $(L,M)$-DTS deferred 
to Section \ref{Generalized-TDZC} where
it is needed. Classical staircase codes \cite{smith} will
be defined as a special case of generalized 
staircase codes in Section \ref{classical-staircase-recovery}.

\subsection{Code Structure}

As with classical staircase codes \cite{smith}, 
a codeword is an 
infinite sequence $B_0, B_1, B_2, \dots$ of $S\times S$
matrices called \emph{blocks} satisfying certain constraints.  The parameter $S$ is referred to as 
the \emph{sidelength}. 
Matrix entries are from the binary field $\mathbb{F}_2$ for concreteness 
but nonbinary alphabets are straightforwardly accommodated. 
We assume zero-based indexing of matrices with the $(i,j)$th 
entry of a matrix $B$ denoted by $(B)_{(i,j)}$
and the index set $\{0,1,\dots,S-1\}$ denoted by $[S]$.

We begin by specifying a collection of $M+1$ permutations 
\begin{align*}
	\pi_k\colon [S] \times [S] &\longrightarrow [S] \times [S]\\
	(i,j)&\longmapsto \pi_k(i,j)
\end{align*}
indexed by $k \in [M+1]$ where $M$ is termed the \emph{degree parameter}.
In other words, for each ${k \in [M+1]}$, we have an invertible function $\pi_k$ 
which maps any given
$(i,j) \in [S] \times [S]$ to $\pi_k(i,j) \in [S] \times [S]$.
We further define 
\begin{align*}
	\Pi_k\colon \mathbb{F}_2^{S\times S} &\longrightarrow \mathbb{F}_2^{S\times S}\\
	B&\longmapsto \Pi_k(B)
\end{align*} 
to be the invertible function mapping a matrix $B \in \mathbb{F}_2^{S\times S}$ to 
its permuted copy $\Pi_k(B) \in \mathbb{F}_2^{S\times S}$ according to $\pi_k$. In 
particular, the entries of $\Pi_k(B)$ are given by
\begin{equation*}
	(\Pi_k(B))_{(i,j)} = (B)_{\pi_k(i,j)}
\end{equation*}
for all $(i,j) \in [S] \times [S]$.
We will assume that $\Pi_0(B) = B$ or, 
equivalently, that $\pi_0(i,j) = (i,j)$, unless otherwise
stated. This is typically without
loss of generality.

Next, we define \emph{rulers}:
\begin{definition}[Ruler]
	A \emph{ruler} of \emph{order} $M+1$ 
	is an ordered set of $M + 1$ \emph{distinct}
	integers $(d_0,d_1,\dots,d_M)$ referred to as
	\emph{marks}.
	A ruler is \emph{normalized} if
	$0 = d_0 < d_1 < \dots < d_M$.
\end{definition}
We will
assume that all rulers are normalized 
unless otherwise stated. Normalization 
is achieved by sorting followed by 
subtraction of the least element from
each mark.

Finally, we specify a linear, systematic $t$-error-correcting \emph{component code} 
$\mathcal{C} \subseteq \mathbb{F}_2^{(M+1)S}$ of length $(M+1)S$
and dimension $(M+1)S-r$ thus having $r$ parity bits. 
This allows us to provide the following 
definition for a \emph{weak generalized staircase 
code}. The qualifier \emph{weak} will be dropped
once we impose certain conditions on the 
ruler and permutations.

\begin{definition}[Weak generalized staircase code]\label{w-gsc-definition}
	Given all three of: a collection of
	permutations, a ruler, and a component code,
	each as just described, 
	a \emph{weak generalized staircase code} is 
	defined by the constraint that the rows of the  
	$S\times (M+1)S$ matrix 
	\begin{equation}\label{constraint-span}
		\big(
		\Pi_M(B_{n-d_M}) \big\vert
		\Pi_{M-1}(B_{n-d_{M-1}}) \big\vert
		\cdots \big\vert
		\Pi_{1}(B_{n-d_1}) \big\vert
		B_n 
		\big)
	\end{equation}
	belong to $\mathcal{C}$ for all integers $n \geq d_M$ 
	and the initialization condition 
	\begin{equation}\label{init-condition}
		B_0 = B_1 = \cdots = B_{d_M - 1} = 0_{S\times S}
	\end{equation}
	where $0_{S\times S}$ is the all-zero matrix.
\end{definition}

Practically, the first $d_M$ blocks \eqref{init-condition} 
need not be transmitted since they are fixed and 
encoding of the $n$th block is performed recursively
starting at $n = d_M$. In particular,
the $n$th
$S\times(S-r)$ block of information bits 
$B_n^{\mathsf{info}}$
produced by a streaming source is
adjoined to $M$ permuted previously encoded blocks
to form the $S \times ((M+1)S-r)$ matrix
\begin{equation}\label{encoder-span}
	\big(
	\Pi_M(B_{n-d_M}) \big\vert
	\Pi_{M-1}(B_{n-d_{M-1}}) \big\vert
	\cdots \big\vert
	\Pi_{1}(B_{n-d_1}) \big\vert
	B_n^{\mathsf{info}} 
	\big)\hspace{-0.81ex}
\end{equation}
whose rows are encoded using 
systematic encoders for $\mathcal{C}$ to produce $r$ parity
bits for each of the $S$ rows.
These parity bits form an $S\times r$ 
block $B_n^{\mathsf{parity}}$
which is adjoined to the information
block to produce $B_n$ as
\begin{equation*}
	B_n = 
	\big(
		\underbrace{\text{---}\; B_n^{\mathsf{info}}\; \text{---}}_{S \times (S-r)}
		\;\big\vert\;
		\underbrace{\text{--}\; B_n^{\mathsf{parity}}\; \text{--}}_{S \times r}
	\big)
\end{equation*}
which satisfies the constraint
that the rows of \eqref{constraint-span} belong to $\mathcal{C}$.  
Therefore, each transmitted block 
$B_n$ contains $S-r$ information columns
and $r$ parity columns yielding a rate of 
\begin{equation*}
	R_\mathsf{unterminated} = 1-\frac{r}{S}
\end{equation*}
which we qualify as ``unterminated'' so that we can later
account for rate losses due to the strategy used for 
ending a transmission. We further observe that we need $r \leq S$ 
or, equivalently, that the rate of $\mathcal{C}$ is at least $M/(M+1)$.

\begin{figure}[t]
	\centering
	\includegraphics[]{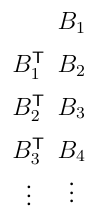}
	\caption{Visualization of a (classical) staircase code \cite{smith}
		which corresponds to $M = 1$, 
		$d_1=1$, and $\Pi_1(B)=B^\mathsf{T}$ (the matrix
		transpose);
		rows belong to $\mathcal{C}$.}\label{sc-example-fig}
\end{figure}
\begin{figure}[t]
	\centering
	\includegraphics[]{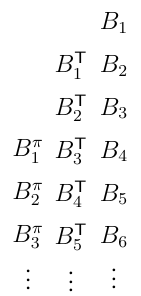}
	\caption{Visualization of a (weak) generalized
		staircase code which corresponds to $M = 2$,
		$(d_1,d_2)=(1,3)$, $\Pi_1(B)=B^\mathsf{T}$ (the matrix
		transpose), 
		and $\Pi_2(B)=B^\pi$ (some permutation);
		rows belong to $\mathcal{C}$.}\label{gsc-example-fig}
\end{figure}

We provide a visualization of the resulting codes in 
Fig.~\ref{sc-example-fig} and Fig.~\ref{gsc-example-fig}
for the cases of $M = 1$ and $M = 2$ respectively 
in the
style of \emph{zipper codes} \cite{sukmadji}. In
these pictures, the rightmost column represents the 
blocks that are actually transmitted while the columns
to its left are delayed and permuted copies that are
used to constrain the transmitted blocks.

It remains to specify precisely what the permutations
and ruler should be. Before proceeding, we provide
a graph perspective which formalizes this question.
This can be skimmed on a first reading.

\subsection{Graph Perspective and Scattering}

For mathematical convenience, especially in the later parts of this paper, 
we extend codewords to be
bi-infinite sequences 
\begin{equation*}
	\dots,B_{-2},B_{-1},B_{0},B_{1},B_{2},\dots
\end{equation*} 
where we define 
\begin{equation*}
	\cdots = B_{-3} = B_{-2} = B_{-1} = 0_{S\times S}
\end{equation*}
so that the constraint on the rows of \eqref{constraint-span}
in a weak generalized staircase code can be presumed to hold for
all integers $n \in \mathbb{Z}$ rather than only $n \geq d_M$.
Practically, nothing changes since the all-zero initialization
blocks are not transmitted and only a finite number of past blocks
are ever needed for encoding.

We define an infinite bipartite graph whose disjoint and independent 
vertex sets are $\mathcal{P} = \mathbb{Z} \times [S] \times [S]$
and $\mathcal{L} = \mathbb{Z} \times [S]$. 
An element $(n,i,j)$ 
of $\mathcal{P}$ corresponds to the 
$(i,j)$th variable in the $n$th block $(B_n)_{(i,j)}$ while
an element $(n',i')$ of $\mathcal{L}$ corresponds to the
component codeword constraint on the $i'$th row of \eqref{constraint-span}
with $n = n'$. In particular, the edge set 
$\mathcal{I}\subseteq \mathcal{P}\times \mathcal{L}$
is given by
\[
	\mathcal{I} = \big\{
	((n,i,j),(n',i')) \in \mathcal{P}\times \mathcal{L} \;\big|\; \\	
	n = n'-d_k, k \in [M+1],
	(i,j) = \pi_k(i',j'), j'\in [S]
	\big\}\text{.}
\]
In other words, this is the (generalized)
\emph{Tanner graph} \cite{Tanner} 
of the code with $\mathcal{P}$ being the variable
nodes and $\mathcal{L}$ being generalized constraint nodes
corresponding to membership in the component code $\mathcal{C}$. 
The variable nodes partition naturally into \emph{blocks} 
where
the $n$th block comprises $S^2$ variable nodes corresponding 
to $(n,i,j)\in \mathcal{P}$ for every $i,j \in [S]$. Moreover, 
the
constraint nodes partition naturally into \emph{constraint
spans} where the $n'$th constraint span corresponds
to the $S$ rows of \eqref{constraint-span} with $n=n'$, i.e., contains $S$ constraint
nodes corresponding to $(n',i') \in \mathcal{L}$ for every $i'\in [S]$.

We denote by $\mathcal{N}(\mathbf{p})$ the neighbourhood
of a given $\mathbf{p}\in\mathcal{P}$ and similarly by 
$\mathcal{N}(\mathbf{q})$ the neighbourhood of a given
$\mathbf{q}\in \mathcal{L}$. 
Observe that we have for all $(n',i') \in \mathcal{L}$ 
\begin{equation}
\label{component-codeword}
	\mathcal{N}((n',i')) =
	\{(n,i,j) \in \mathcal{P} \mid 
	n = n'-d_k, k\in[M+1], \\
	 (i,j) = \pi_k(i',j'),  j'\in [S]
	\}
\end{equation}
with 
\begin{equation*}
	\abs{\mathcal{N}((n',i'))} = (M+1)S\text{,}
\end{equation*}
i.e., that each constraint involves $(M+1)S$
variables which are a row of \eqref{constraint-span}.
Moreover, for all $(n,i,j)\in \mathcal{P}$, we have
\begin{equation}
\label{bit-neighbours}
	\mathcal{N}((n,i,j)) =
	\{(n',i') \in \mathcal{L} \mid 
	n' = n + d_k, k \in [M+1], \\
	(i',j')=\pi_k^{-1}(i,j), j'\in [S]
	\}
\end{equation}
with 
\begin{equation}\label{bit-degree}
	\abs{\mathcal{N}((n,i,j))} = (M+1)\text{,}
\end{equation}
i.e., that each variable is involved in
$M+1$ constraints.

While each coded 
symbol is protected by $M+1$ component codewords
as per \eqref{bit-degree}, 
we have no guarantee that our error correcting
capability improves with increasing 
$M$ unless there is some amount of disjointness
between the component codewords.
One way to guarantee this is by imposing 
the property that
our Tanner graph is \emph{$4$-cycle-free}.  
The same property is referred
to as \emph{scattering} in the zipper code 
framework \cite{sukmadji} and as
\emph{self-orthogonality} in the context of threshold decoding 
\cite{CSOC}.

\begin{definition}[Scattering/$4$-cycle-freeness/self-orthogonality]
	A code with Tanner graph $(\mathcal{P},\mathcal{L},\mathcal{I})$
	is \emph{scattering} if, for all
	$\mathbf{q}_1, \mathbf{q}_2\in\mathcal{L}$
	with $\mathbf{q}_1\neq\mathbf{q}_2$, we have
	\begin{equation}\label{scattering}
		\abs{\mathcal{N}(\mathbf{q}_1)\cap \mathcal{N}(\mathbf{q}_2)}
		\leq 1\text{,}
	\end{equation}
	i.e., that pairs of distinct (generalized) constraints share
	at most one variable.
\end{definition}
  
It is well-known that $4$-cycle-freeness
of the Tanner graph of a code has performance
implications the specifics of which depend on the
channel or error model under consideration
as well as the decoding method. In our context, 
the relevant consequence is as follows:

\begin{proposition}\label{stall-pattern-bound}
	For any code with $4$-cycle-free Tanner
	graph $(\mathcal{P},\mathcal{L},\mathcal{I})$ having  
	degree $M+1$ variable nodes $\mathcal{P}$ along with
	generalized constraint nodes $\mathcal{L}$ corresponding to 
	membership in a $t$-error-correcting component code,
	the Hamming weight of an uncorrectable 
	error under iterative
	bounded distance decoding of the constraints is at least $(M+1)t+1$.
\end{proposition}
\begin{proof}
	A minimum weight uncorrectable error corresponds to
	a finite subset of the
	variables $\mathcal{V} \subseteq \mathcal{P}$
	such that any involved constraint in $\mathcal{L}$
	is involved in at least $t+1$ of these variables 
	(else, the weight of the error 
	could be further reduced).
	Given any such error, we can
	define a finite simple graph 
	$\mathcal{G} = (\mathcal{V}, \mathcal{E})$
	in which two vertices are adjacent 
	if and only if they
	have a common neighbour in the Tanner graph
	$(\mathcal{P},\mathcal{L},\mathcal{I})$.
	A fixed vertex $u\in \mathcal{V}$ then
	has $M+1$ neighbours in $(\mathcal{P},\mathcal{L},\mathcal{I})$ 
	each of which contributes at least $t$ neighbours 
	to $u$ in $\mathcal{G}$. Moreover, these contributions
	must be disjoint otherwise $4$-cycle-freeness is
	violated. The degree of $u$ in $\mathcal{G}$ 
	is then at least $(M+1)t$ implying that
	$\abs{\mathcal{V}} \geq (M+1)t+1$. 
\end{proof}

Note that this bound is usually---but not always---quite loose and does
not address the multiplicity of uncorrectable error patterns. 
Its purpose is solely to provide a theoretical guarantee
for asymptotic control over the error floor rather than to estimate
the error floor or prescribe parameters.

Our goal is then to find 
a ruler and permutations which preserve 
the scattering property \eqref{scattering}
which is indeed possessed by classical
staircase codes \cite{smith}.
We first observe that \eqref{scattering} already holds
for distinct constraints within the same constraint span since
they correspond to different rows of \eqref{constraint-span}.
In particular, for all $(n',i'_1),(n',i'_2) \in \mathcal{L}$ with
$i'_1 \neq i'_2$, we have 
\begin{equation*}
	\mathcal{N}((n',i'_1))\cap \mathcal{N}((n',i'_2)) = \emptyset\text{.}
\end{equation*}
Therefore, we can restrict our interest to distinct pairs of 
constraints $(n'_1,i'_1),(n'_2,i'_2)\in \mathcal{L}$ 
corresponding to distinct constraint spans $n'_1 \neq n'_2$
and $i'_1,i_2' \in [S]$ equal or unequal. Consider now the problem of constructing a point
of intersection between these constraint spans $(n,i,j)\in \mathcal{P}$
where
\begin{equation}\label{intersection}
	(n,i,j)\in \mathcal{N}((n'_1,i'_1))\cap \mathcal{N}((n'_2,i'_2))
\end{equation}
and $n'_2 > n'_1$ without loss of generality.
From \eqref{component-codeword}, we see that we must
firstly have that 
\begin{equation*}
	n = n'_1 - d_{k_1} = n'_2 - d_{k_2}
\end{equation*}
thus 
\begin{equation}\label{inter-block-condition}
	n'_2 - n'_1 = d_{k_2} - d_{k_1} > 0
\end{equation}
for some $k_1,k_2\in [M+1]$ with $k_2 > k_1$.
Secondly, given such a solution, we must further
have that
\begin{equation}\label{intra-block-condition}
	(i,j) = \pi_{k_1}(i'_1,j'_1) = \pi_{k_2}(i'_2,j'_2) 
\end{equation}
for some $j'_1, j'_2 \in [S]$. 

A natural sufficient condition for guaranteeing scattering
is then by a choice of ruler so that \eqref{inter-block-condition}
has a unique solution $k_1,k_2 \in [M+1]$ if any and a choice of permutations
so that \eqref{intra-block-condition} has a unique solution $j'_1, j'_2 \in [S]$.
Such solutions determine and thus guarantee the uniqueness of $(n,i,j)$
satisfying \eqref{intersection}. 
We refer to the property of 
having unique solutions
to \eqref{inter-block-condition} 
as \emph{inter-block scattering} since
it is the property that distinct constraint
spans share at most one block. 
On the other hand, we refer to the property of having unique
solutions to \eqref{intra-block-condition}
as \emph{intra-block scattering} since it
is the scattering property restricted
to variables within a block. Intra-
and inter-block scattering together
imply scattering. We now proceed to
specify rulers and permutations
leading to scattering.

\subsection{Intra-Block Scattering Permutations via Finite-Geometric Nets}

It is both practically convenient and mathematically
sufficient to consider permutations
defined by linear-algebraic operations on matrix indices $(i,j) \in 
[S]\times [S]$.
To do so, we must associate the elements of 
the index set $[S]$ with the elements of 
a finite commutative ring 
$\mathcal{R}$ of cardinality $S$.
A permutation $\pi\colon [S] \times [S] \longrightarrow [S] \times [S]$
can then be interpreted as a linear function
$\pi\colon \mathcal{R} \times 
\mathcal{R} \longrightarrow \mathcal{R} \times \mathcal{R}$
defined by an invertible $2\times 2$ matrix as
\begin{equation}\label{linear-algebraic-permutation}
	\pi(i,j) =
	(i,j) \cdot
	\begin{pmatrix}
		a & b \\
		c & d 
	\end{pmatrix}
	= 
	(ai + cj, bi + dj)
\end{equation}
where $a,b,c,d \in \mathcal{R}$ and 
$(ad-bc)^{-1}$ exists in $\mathcal{R}$.
The inverse function is then immediately 
obtained via the
matrix 
\begin{equation*}
	(ad-bc)^{-1}
	\begin{pmatrix}
		d & -b \\
		-c & a 
	\end{pmatrix}\text{.}
\end{equation*}

As will be seen shortly, intra-block scattering families of 
permutations are most easily constructed when $\mathcal{R}$
has sufficiently many elements with invertible differences.
Finite fields are an obvious choice since 
any pair of distinct elements has an invertible
difference. 
We can then take $\mathcal{R} = \mathbb{F}_q$, the
finite field of order $q$, with $S = q$ a prime
power and associate $[S]$ with 
$\{0, 1, \alpha, \alpha^2, \dots, \alpha^{S-2}\}$ 
where $\alpha$ is a primitive element of $\mathbb{F}_q$.
However, a convenient alternative option is to take
$\mathcal{R} = \mathbb{Z}_S$, the ring of integers modulo $S$,
where $S$  is chosen to have a sufficiently
large \emph{least prime factor} denoted $\mathsf{lpf}(S)$. 
This is because $\{0,1,\dots,\mathsf{lpf}(S)-1\}$ 
have invertible differences modulo $S$.
In particular, we have the following remark:
\begin{remark}\label{modular-invertibility-remark}
	For any positive integer $S$ with least prime factor $\mathsf{lpf}(S)$,
	the integers $\{\pm 1, \pm 2,\dots, \pm (\mathsf{lpf}(S)-1)\}$ 
	are coprime with $S$ and are thus multiplicatively invertible
	modulo $S$.
\end{remark}

We now proceed to define \emph{$(M+1,S)$-nets}
which are well-studied objects in combinatorics and finite
geometry with close connections or equivalences to orthogonal arrays, mutually
orthogonal Latin squares, transversal designs, and projective planes
as discussed in~\cite{Colbourn}.

\begin{definition}[$(M+1,S)$-net {\cite[Part~III:~Ch.~3]{Colbourn}}]\label{net-definition}
	An $(M+1,S)$-net is a set of $S^2$ elements referred to as \emph{points} 
	along with a collection of $M+1$ partitions of these points into subsets 
	of size $S$ called \emph{lines} such that distinct lines intersect in at 
	most one point. 
\end{definition}

We next make the key observation that $(M+1,S)$-nets 
are equivalent to intra-block scattering permutations:

\begin{remark}
	If the (variable-valued) entries of an $S\times S$ matrix $B$ are the points, and the rows of the $M+1$
	permuted copies $\Pi_k(B)$ for $k \in [M+1]$ are the $M+1$ partitions of these points into lines,
	then an $(M+1,S)$-net is equivalent to a collection of permutations
	of $[S]\times [S]$ with the following property: 
	For any distinct $k_1,k_2\in [M+1]$, 
	any row of $\Pi_{k_1}(B)$ has a single element in common with any row
	of $\Pi_{k_2}(B)$. Equivalently, for fixed distinct $k_1,k_2 \in [M+1]$
	and fixed distinct $i'_1,i'_2\in [S]$, the equation
	$\pi_{k_1}(i'_1,j'_1) = \pi_{k_2}(i'_2,j'_2)$ has 
	a unique solution for $j'_1, j'_2\in [S]$. 
\end{remark}

A large number of constructions and open
problems are associated with $(M+1,S)$-nets
and can be found in~\cite{Colbourn}. 
Fortunately, for our practically-oriented purposes, 
linear-algebraic permutations 
over conveniently chosen rings already provide enough flexibility
and are trivial to characterize. In particular, 
we have the following: 

\begin{proposition}[Complete characterization of linear-algebraic permutations
	defining $(M+1,S)$-nets]\label{intra-block-scattering-theorem}
A collection of $M+1$ permutations of $\mathcal{R} \times \mathcal{R}$, 
where $\mathcal{R}$ is a finite commutative ring of cardinality S, defined by
a collection of invertible $2 \times 2$ matrices as in \eqref{linear-algebraic-permutation}, 
define an $(M+1,S)$-net if and only if, for any pair of distinct matrices $A, \tilde A$ in the collection
where 
\begin{equation*}
	A = 
	\begin{pmatrix}
		a & b \\
		c & d
	\end{pmatrix},\;
	\tilde A = 
	\begin{pmatrix}
		\tilde a & \tilde b \\
		\tilde c & \tilde d
	\end{pmatrix}\text{,}
\end{equation*} 
we have that $c\tilde d - d \tilde c$ is invertible 
in $\mathcal{R}$.
\end{proposition}
\begin{proof}
	This follows by imposing that the 
	linear system $(i,j)A=(\tilde i,\tilde j)\tilde{A}$
	has a unique solution for $(j,\tilde j)$ 
	given fixed row indices $(i,\tilde i)$
	by rearranging and imposing an invertible determinant.
\end{proof}

We now provide some concrete examples.

\begin{example}\label{net-perms}
	If $\mathcal{R}=\mathbb{Z}_S$ and $M \leq \mathsf{lpf}(S)$, then
	the identity permutation $I_{2\times 2}$
	together with 
	\begin{equation*}
		\begin{pmatrix}
			0 & 1\\
			1 & z
		\end{pmatrix}
	\end{equation*}
	for $z \in \{0,1,\dots, M-1\}$ define an $(M+1,S)$-net. 
	This can be easily verified by applying Proposition \ref{intra-block-scattering-theorem} 
	together with Remark \ref{modular-invertibility-remark}.
	Concretely, we can take $\pi_0(i,j) = (i,j)$ and 
	define for $k \in \{1,2,\dots,M\}$
	\begin{equation*}
		\pi_k(i,j) =
		(j, i+(k-1)j)
		\pmod S 
	\end{equation*}
	where $M \leq \mathsf{lpf}(S)$.
\end{example}

\begin{example}\label{net-perms-invo}
	If $\mathcal{R}=\mathbb{Z}_S$ and $M \leq \mathsf{lpf}(S)$, then
	the identity permutation $I_{2\times 2}$
	together with the involutions
	\begin{equation*}
		\begin{pmatrix}
			-z & 1-z^2 \\
			1 & z
		\end{pmatrix}
	\end{equation*}
	for $z \in \{0,1,\dots, M-1\}$ define an $(M+1,S)$-net. 
	This is verified by noting that Proposition \ref{intra-block-scattering-theorem}
	is indifferent to the first rows of the matrices so long as they are invertible.
	Concretely, we take $\pi_0(i,j) = (i,j)$ and 
	define for $k \in \{1,2,\dots,M\}$
	\begin{equation*}
		\pi_k(i,j) = (-(k-1)i+j, (1-(k-1)^2)i + (k-1)j) \hspace{-1ex} \pmod S
	\end{equation*}
	where $M \leq \mathsf{lpf}(S)$.
\end{example}

\begin{example}\label{net-perms-finite-field}
	If $\mathcal{R}=\mathbb{F}_S$ with $S = q$ a prime power and 
	$M\leq S$,
	the identity permutation $I_{2\times 2}$
	together with
	\begin{equation*}
		\begin{pmatrix}
			0 & 1\\
			1 & z
		\end{pmatrix}
	\end{equation*}
	for  $z\in \{0, 1, \alpha, \alpha^2, \dots, \alpha^{M-2}\}$ with $\alpha$ a primitive element of $\mathbb{F}_S$ form an $(M+1,S)$-net.
\end{example}

We emphasize that in Examples \ref{net-perms} and
\ref{net-perms-invo}, we must 
choose $S$ so that $M\leq \mathsf{lpf}(S)$ 
otherwise an $(M+1,S)$-net is not formed. 
If $S = p$ a prime number, then
this condition becomes that $M\leq S$. 
If $M = S$ is required, then Example \ref{net-perms-finite-field}
provides more flexibility in the admissible
values of $S$ which can be a prime power. 

\subsection{$(M+1,1)$-nets}

It will be useful later to have
a family of trivial nets for the case of $S = 1$ 
for all $M \geq 0$. In particular, we interpret the collection of 
partitions in Definition \ref{net-definition} as
allowing for repeated copies of the same partition, 
or equivalently, repeated copies of the same
corresponding permutation. In the case
of $S > 1$, the definition is unchanged since
repeated copies are automatically excluded since
they necessarily lead to lines intersecting in 
$S > 1$ points. In the case of $S=1$, we admit
$M+1$ copies of the trivial partition of
the singleton set into itself leading to lines
which intersect in one point by virtue of only
containing that point. Alternatively, we define 
this case explicitly:
\begin{definition}[$(M+1,1)$-net]\label{trivial-net-family}
	An $(M+1,1)$-net comprises $M+1$ copies of the trivial partition of a singleton set
	into itself. The corresponding permutations of $[1]\times [1]$ are $M+1$ copies
	of the identity permutation $\pi_k(i,j)=(i,j)$ for $k\in [M+1]$ and
	$(i,j)\in [1]\times [1]$, i.e., $(i,j) = (0,0)$.
\end{definition}

\subsection{Inter-Block Scattering via Golomb Rulers}

Inter-block scattering rulers are equivalent to
\emph{Golomb} rulers defined as follows:

\begin{definition}[Golomb ruler {\cite[Part~VI:~Ch.~19]{Colbourn}}]
	An order $M+1$ ruler $(d_0,d_1,\dots,d_M)$ is a 
	\emph{Golomb} ruler if all
	differences between distinct marks 
	$d_{k_1}-d_{k_2}$ 
	where $k_1,k_2\in [M+1]$ and $k_1\neq k_2$, 
	are distinct.
	The \emph{length} of a Golomb ruler is its largest such difference.
\end{definition}

Unless otherwise stated, we will assume that all Golomb 
rulers are normalized so that $0 = d_0 < d_1 < \dots < d_M$.
The length of a Golomb ruler is then simply its largest mark $d_M$.
The \emph{set of distances} of a Golomb ruler 
is the set of its positive differences
\begin{equation*}
	\{d_{k_2}-d_{k_1} \mid k_2 > k_1, k_1,k_2\in [M+1]\}
\end{equation*}
which is necessarily of cardinality $\binom{M+1}{2} = \frac{(M+1)M}{2}$
by the distinctness of these differences.
This implies that the length $d_M$ of an order $M+1$ 
Golomb ruler satisfies the \emph{trivial bound}
\begin{equation}\label{Golomb-length-lower-bound}
	d_M \geq \frac{(M+1)M}{2}
\end{equation}
with 
equality if and only if its set of distances
is
precisely $\{1,2,\dots,(M+1)M/2\}$. 
A Golomb ruler for which the trivial
bound
holds with equality is said to be \emph{perfect}.

\begin{remark}
	Given a Golomb ruler, 
	we see that \eqref{inter-block-condition} 
	has a solution whenever $n'_2 - n'_1$ 
	is in its set of distances and that the 
	solution is unique due to the absence of repeated distances.
\end{remark}

Next, observe that Golomb rulers of every
order exist by taking, for example,
$d_k = 2^k-1$ for $k\in[M+1]$ in which the distance 
between consecutive marks is larger than the 
distances between all previous marks. 
This results
in an exponential length $d_M = 2^M-1$
which is very far from the quadratic lower
bound \eqref{Golomb-length-lower-bound}.
Indeed, this is a problem since we can
see from \eqref{encoder-span} that 
our encoding memory is 
proportional to the ruler length $d_M$.
Therefore, we prefer Golomb rulers which
have the smallest possible length
$d_M$ for a given order $M+1$. 
Such Golomb rulers 
are referred to 
as \emph{optimal} and are well-studied. 

We provide for the reader's convenience 
optimal Golomb rulers of order up to
$M+1 = 10$ in Table \ref{optimal-Golomb-ruler-table}.
Observe that the first four optimal Golomb rulers
are verified to be optimal by seeing that they
are perfect.
Beyond these cases, optimal Golomb rulers are typically
found by computer search.
Optimal Golomb rulers of order up to 
$M+1 = 20$ are compiled in \cite[Part~VI:~Ch.~19]{Colbourn}.
While it is an open problem to exhibit optimal
Golomb rulers of every order, 
Golomb rulers of length $d_M < (M+1)^2$
have been found for all $M+1 < 65000$ in \cite{Apostol-Golomb}.
Moreover, an explicit construction of Golomb rulers 
with cubic length in arbitrary order $M+1$ 
can also be found in \cite{Apostol-Golomb}.

\begin{table}[t]
	\centering \caption{Some optimal Golomb rulers from {\cite[Part~VI:~Ch.~19]{Colbourn}}}
	\begin{tabular}{c|c|l}\label{optimal-Golomb-ruler-table}
		order & length & optimal Golomb ruler \\
		\hline
		1 & 0 & 0\\
		2 & 1 & 0 1\\
		3 & 3 & 0 1 3\\
		4 & 6 & 0 1 4 6\\
		5 & 11 & 0 1 4 9 11\\
		6 & 17 & 0 1 4 10 12 17\\
		7 & 25 & 0 1 4 10 18 23 25\\
		8 & 34 & 0 1 4 9 15 22 32 34\\
		9 & 44 & 0 1 5 12 25 27 35 41 44\\
		10 & 55 & 0 1 6 10 23 26 34 41 53 55 
	\end{tabular}
\end{table}

\subsection{Recovery of Classical Staircase Codes $(L=1,M=1)$}\label{classical-staircase-recovery}

Consider taking $M = 1$ with the optimal Golomb ruler
of order $M+1 = 2$ and using the family of permutations given 
in Example \ref{net-perms} or in Example \ref{net-perms-invo}.
We get that $d_1 = 1$ and that $\pi_1(i,j) = (j,i)$
or, equivalently, that $\Pi_1(B) = B^\mathsf{T}$ (the matrix
transpose). Moreover, for $S > 1$, the requirement
$M \leq \mathsf{lpf}(S)$ holds trivially since 
$\mathsf{lpf}(S) \geq 2 \geq M$. We therefore
recover the classical staircase codes of 
\cite{smith} exactly. This motivates the following
definition for a bona fide generalized staircase code:
 
\begin{definition}[Generalized staircase code]\label{gsc-definition}
	A \emph{generalized staircase code} is a weak generalized 
	staircase code as in Definition \ref{w-gsc-definition}
	with the added properties that the permutations
	form an $(M+1,S)$-net and that the ruler is an
	order $M+1$ Golomb ruler.
\end{definition}

\section{Extended-Hamming-Based Generalized Staircase Codes}
\label{instantiation-section}

Staircase codes and related product-like codes 
\cite{sukmadji} typically use a $t$-error-correcting
Bose--Ray-Chaudhuri--Hocquenghem (BCH) component code
$\mathcal{C}$ to protect each coded bit with 
$M + 1 = 2$ component codewords. 
This typically requires that
$t \geq 3$ to achieve error floors below $10^{-15}$
\cite{Low-Rate-Hardware-Staircase} which imposes a limit
on our ability to improve the energy-efficiency of such 
schemes since the energy cost of a BCH decoding 
is roughly $\mathcal{O}(t^2)$ \cite{BCH-Power-Estimate}.
However, Proposition \ref{stall-pattern-bound} suggests that for
our generalized staircase codes which are scattering, 
we can reduce $t$ without increasing the error floor
by compensating with a larger bit degree $M+1$.
Accordingly, we consider the use of generalized staircase
codes with single-error-correcting, i.e., $t=1$, components.
Reasoning heuristically, 
we expect the resulting schemes to be more 
energy-efficient than 
those employing a smaller $M+1$ with a larger $t$
since the energy cost per corrected bit would be 
$\mathcal{O}(t^2/t) = \mathcal{O}(t)$ with the lower bound
of Proposition \ref{stall-pattern-bound} held \emph{constant}.

We consider the use of single-error-correcting, 
double-error-detecting extended Hamming codes
which have length $2^{r-1}$ and dimension $2^{r-1}-r$.
Such codes have the $r\times 2^{r-1}$ parity-check matrix 
with 
columns given by the binary representations
of $2j+1$ for every $j \in [2^{r-1}]$:
\begin{equation}\label{non-systematic-Hamming}
	\begin{pmatrix}
		0 & 0 & 0 & 0 & \cdots & 1\\
		0 & 0 & 0 & 0 &\cdots & 1 \\
		\vdots & \vdots & \vdots & \vdots &\ddots & \vdots \\
		0 & 0 & 0 & 0 &\cdots & 1\\
		0 & 0 & 1 & 1 &\cdots & 1\\
		0 & 1 & 0 & 1 &\cdots & 1\\
		1 & 1 & 1 & 1 &\cdots & 1
	\end{pmatrix}\text{.}
\end{equation}
Famously, syndrome decoding of such codes is trivial with
the first $r-1$ bits of the syndrome interpreted as 
an integer directly giving the error location $j \in [2^{r-1}]$.

Unfortunately, \eqref{non-systematic-Hamming} 
must in general have its columns permuted in order 
to have a corresponding systematic generator matrix. 
This 
would then require that both decoding and generation of 
columns of the parity-check matrix is done by lookup table
(LUT). Both of these operations must be done repeatedly
during the iterative decoding of generalized staircase codes.
It is thus useful, both for efficient software-based simulation
and for hardware implementation, to be able 
to perform these operations as simply as in the case of
\eqref{non-systematic-Hamming} unpermuted. We provide
some solutions to this problem for software and 
hardware respectively in 
Sections \ref{systematic-Hamming-section} and \ref{systematic-Hamming-section-hardware}.

\subsection{Simulations}

We consider a component code 
$\mathcal{C}$ of length $(M+1)S$ and dimension $(M+1)S-r$
which is obtained by shortening a parent extended Hamming
code of length $2^{r-1}$ 
in the first $2^{r-1} - (M+1)S$ positions
where 
${r-1 = \lceil \log_2((M+1)S) \rceil}$.
The parent codes are as defined 
by Table \ref{systematic-Hamming-table} later.
Furthermore, the permutations are taken
to be those of Example \ref{net-perms-invo} thus
requiring that $M \leq \mathsf{lpf}(S)$,
and the rulers are taken to be the optimal
Golomb rulers from Table \ref{optimal-Golomb-ruler-table}.
We then assume transmission across a 
binary symmetric channel (BSC) and
perform sliding window decoding of 
$W$ consecutive blocks at a time  where one iteration 
comprises decoding all constraints in the window consecutively and $I$ iterations are performed
before advancing the window by one block
similarly to \cite{smith}.
To facilitate simulation of the waterfall regime 
where error propagation events dominate, 
we convert to a
block code via a termination-like strategy in which a
frame is formed from $F$ consecutive blocks. For the
last $W$ of these $F$ blocks, the information
portion is presumed to be zero and only the parity
portion is transmitted with the encoder and decoder state
being reset between consecutive frames. This results in a 
block code rate of 
\begin{equation*}
	R = \frac{(S-r)(F-W)}{S(F-W)+Wr}
\end{equation*}
which approaches $R_\mathsf{unterminated} = 1-r/S$ as $F\to \infty$.
The reader is referred to \cite{sukmadji} for alternative
decoding and termination-like procedures for similar families of codes.

\begin{figure}[t]
	\centering
	\includegraphics[width=0.5\columnwidth]{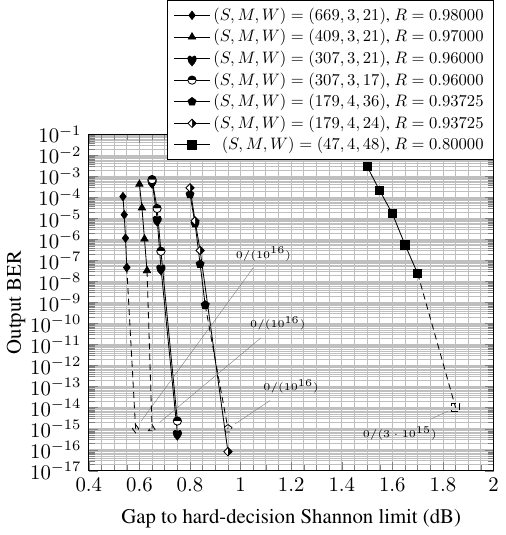}
	\caption{Simulation results for extended-Hamming-based generalized staircase
		codes with parameters given in Table \ref{sim-param-table}.}\label{sim-curves}
\end{figure}

\begin{table}[t]
	\centering \caption{Simulation parameters for Fig.~\ref{sim-curves}.}
	\begin{tabular}{c|c|c|c|c|c|c|c}
		$S$ & $M$ & $F$ & $W$ (Mbits) & $R$ & $I$ \ & gap (dB) & in.\ BER\\
		\hline
		669 & 3 &  725  & 21 (9.40) & 0.980 & 3 & 0.585 & 9.86e-4\\
		409 & 3 &  926  & 21 (3.51) & 0.970 & 3 & 0.650 & 1.57e-3\\
		307 & 3 &  885  & 21 (1.98) & 0.960 & 4 & 0.750 & 2.09e-3\\
		307 & 3 &  717  & 17 (1.60) & 0.960 & 4 & 0.750 & 2.09e-3\\
		179 & 4 & 1634  & 36 (1.15) & 0.937 & 4 & 0.950 & 3.25e-3\\
		179 & 4 & 1089  & 24 (0.77) & 0.937 & 4 & 0.950 & 3.25e-3\\
		47 & 4 &  912  & 48 (0.11) & 0.800 & 6 & 1.850 & 1.05e-2
	\end{tabular}\label{sim-param-table}
\end{table}

We developed a highly-optimized software simulator
for the BSC performance
of the resulting extended-Hamming-based 
generalized staircase codes. This simulator
is available online \cite{decsim-code}
and is capable of achieving simulated 
throughputs of several gigabits
per second per core of a modern multi-core consumer processor 
which allows for direct verification of sub-$10^{-15}$
error floors. This is enabled both by the simplicity of syndrome
decoding of Hamming components and 
the fact that decoder behaviour
depends only on the syndromes of the received component codewords.
Therefore, simulation can be performed directly in syndrome domain with individual
syndromes represented by individual integers.
Performing Bernoulli sampling of errors for a BSC with crossover
probability $\epsilon$ as geometrically-distributed jumps between
error locations leads to roughly $\epsilon S^2 (M+1)$ primitive 
operations being needed to update the decoder state in response 
to the reception of a new block. Estimating crudely and heuristically,
one decoding iteration similarly takes roughly 
$\epsilon W S^2 (M+1)$ primitive operations.

A few representative design examples are considered 
with parameters given in Table \ref{sim-param-table} and the 
corresponding simulation results given in Fig.~\ref{sim-curves}
as bit error rate (BER) curves.
We parameterize the BSC by the equivalent gap to the
hard-decision Shannon limit under binary phase-shift keying (BPSK) transmission 
on an additive white Gaussian noise (AWGN) channel
at the respective rates $R$ of the codes considered, 
rather than by the crossover probability. 
The last two rows of Table \ref{sim-param-table}
provide the gap to the hard-decision Shannon limit
along with the corresponding 
BSC crossover probability or input BER 
for the last point of each curve, i.e., the point of least BER.
Where zero bit errors were measured, a dashed line is 
used in Fig.~\ref{sim-curves} to indicate the simulated operating point
with the number of error-free bit transmissions also indicated. 
For all but the last code considered,
we simulate the transmission of at least $10^{16}$
bits at the last plotted operating point, providing
strong confidence that BERs are below $10^{-15}$.

Relative to codes using two triple-error-correcting BCH components
per bit in \cite{sukmadji,smith,Low-Rate-Hardware-Staircase},
the example codes perform roughly $0.2$ to $0.5$ dB worse
in terms of the gap to the hard-decision Shannon limit
as the code rate ranges from $0.98$ down to $0.8$.
In return, we get reduced complexity and latency.
Some heuristic estimates of performance, latency, and complexity
in comparison to classical staircase codes 
\cite{smith} are provided in Tables \ref{complexity-comparison-table-094}
and \ref{complexity-comparison-table-096}. In high-throughput contexts, 
latency is directly proportional to the number of bits
in the decoding window $WS^2$. On the other hand, the number
of component codewords in the decoding window is about $WS$. 
In practice, only a fraction of these must be decoded each iteration
and we assume this fraction to be constant for simplicity.
Moreover, we assume that for the relatively small values of $M$
considered, the cost
of each component code decoding is dominated by the cost of invoking the 
component code decoder $t^2$ rather than the cost of updating
the syndromes of codewords involved in each corrected bit.
This leads to a heuristic complexity measure of $IWSt^2$.

\begin{table}[t]
	\centering \caption{Heuristic estimates of performance, latency, and complexity for generalized staircase codes with $R \approx 0.937$}
	\begin{tabular}{r|c|c}
		& staircase \cite{smith,Dmitri-Hardware-Staircase} & this work \\\hline
		gap to Shannon at 1e-15 & 0.56 dB & 0.95 dB \\
		latency: $WS^2$ & 1.8e6 & 7.7e5 \\
		iterations: $I$ & 3 & 4 \\
		decodings per iteration: $WS$ & 3.6e3 & 4.3e3 \\
		cost per decoding: $t^2$ & 9 & 1 \\
		complexity score: $IWSt^2$ & 9.7e4 & 1.7e4
	\end{tabular} \label{complexity-comparison-table-094}
\end{table}

\begin{table}[t]
	\centering \caption{Heuristic estimates of performance, latency, and complexity for generalized staircase codes with $R \approx 0.96$}
	\begin{tabular}{r|c|c}
		& staircase \cite{smith,sukmadjithesis} & this work \\\hline
		gap to Shannon at 1e-15 & 0.49 dB & 0.75 dB \\
		latency: $WS^2$ & 4e6 & 1.6e6 \\
		iterations: $I$ & 4 & 4 \\
		decodings per iteration: $WS$ & 5e3 & 5.2e3 \\
		cost per decoding: $t^2$ & 9 & 1 \\
		complexity score: $IWSt^2$ & 18e4 & 2e4
	\end{tabular} \label{complexity-comparison-table-096}
\end{table}

\subsection{LUT-Free Decoding of Systematic Hamming Codes with Affine Permutations in $\mathbb{Z}_{2^{r-1}}$}\label{systematic-Hamming-section}

We now revisit the problem of constructing systematic Hamming codes
that admit a LUT-free decoding. This amounts to finding a permutation
of the columns of \eqref{non-systematic-Hamming} such that a corresponding
systematic generator matrix exists \emph{and}
both the forward and inverse permutations are 
easily calculated without a LUT. Note that the
existence of a systematic generator matrix 
is equivalent to the invertibility of the 
$r \times r$ matrix formed by the last $r$
columns of the parity-check matrix. We will refer
to a permutation of the columns of \eqref{non-systematic-Hamming}
achieving this as a \emph{systematizing} permutation.

A natural approach to this problem is to consider 
affine permutations in $\mathbb{Z}_{2^{r-1}}$.
In particular, we interpret the column index set 
$[2^{r-1}]$ as the ring of integers modulo $2^{r-1}$, 
i.e., $\mathbb{Z}_{2^{r-1}}$.
All odd integers are invertible modulo $2^{r-1}$ so
that any odd
$a \in \mathbb{Z}_{2^{r-1}}$ and any 
possibly non-odd $b\in \mathbb{Z}_{2^{r-1}}$ define
an algebraic permutation 
$\tau\colon \mathbb{Z}_{2^{r-1}}\longrightarrow \mathbb{Z}_{2^{r-1}}$ 
where
\begin{equation*}
	\tau(j) = aj + b \pmod{2^{r-1}}
\end{equation*}
and 
\begin{equation*}
	\tau^{-1}(j) = a^{-1}(j - b) \pmod{2^{r-1}}\text{.}
\end{equation*}
In the context of software-based simulation, 
computing $\tau$ and $\tau^{-1}$ can be faster on a 
modern processor than performing the equivalent table
lookup. By a computer search, values of $a$ and $b$ for 
each $(r-1)\in \{3,4,\dots,16\}$ were found 
such that $\tau$ is systematizing. The resulting values 
are provided in Table \ref{systematic-Hamming-table}.
This means that both decoding a syndrome
as well as generating columns 
of the parity-check matrix can be done as simply as computing 
$\tau^{-1}$ and $\tau$.
Adaptation to a code shortened in the first $s$ positions
is accomplished by replacing $b$ with $b+as$.

We report that, somewhat counter-intuitively, the choice 
of column permutation for the extended Hamming parity-check matrix
leads to a slight but measurable improvement in iterative decoding 
performance even in the absence of shortening. 
This is because the ordering of the 
columns impacts the spatial correlations between errors of
odd weight greater than or equal to three and the consequent 
miscorrections. As a result, the use of the extended Hamming codes
defined by Table \ref{systematic-Hamming-table} leads to
slightly improved performance relative to systematizing permutations
which preserve runs of naturally ordered columns. The same 
performance gain is also realized by a random choice
of systematizing permutation which suggests that this 
gain
is due to a decorrelation effect.

\begin{table}[t]
	\centering \caption{Parameters for systematizing affine permutations for extended Hamming codes}
		\begin{tabular}{l|l|l|l}\label{systematic-Hamming-table}
			$r-1$ & $a$ & $b$ & $a^{-1}$ \\ \hline
			3 &1 &1  &  1 \\
			4 &3 & 0 &  11\\
			5 &3  & 0 & 11 \\
			6 &3  &  3&  43\\
			7 &5  & 5 &  77\\
			8 &9  & 11 & 57 \\
			9 &19  & 19 &  27\\
			10 &27  & 27 & 531 \\
			11 &53  & 53 &  541\\
			12 &89  & 89 &  2025\\
			13 &163  & 170 &  4875\\
			14 &301  & 308 &  13989\\
			15 &553  & 553 & 14873\\
			16 &1065  & 1155 & 55321
		\end{tabular}
\end{table}

\subsection{LUT-Free Decoding of Systematic Hamming Codes with Boolean Functions}\label{systematic-Hamming-section-hardware}

\begin{table}[t]
	\centering\caption{Parameters for systematizing permutations for extended
		Hamming Codes via Boolean Functions}
	\begin{tabular}{r|c|c}\label{boolean-parameters}
		$i$ & $d^{(i)}$ & $\sum_{u=0}^{d^{(i)}-1} s^{(i)}_u2^u$ \\\hline
		2 &2 &0\\
		3 &2 &0\\
		4 &3 &0\\
		5 &3 &1\\
		6 &3 &2\\
		7 &2 &0\\
		8 &4 &0\\
		9 &4 &1\\
		10 &4 &2\\
		11 &4 &3\\
		12 &4 &4\\
		13 &4 &5\\
		14 &4 &6\\
		15 &3 &0
	\end{tabular}
\end{table}

In the context of software implementation, the approach presented
in Section \ref{systematic-Hamming-section} is as good as we can hope
for since modern processors already have dedicated multipliers. On
the other hand, in the context of hardware implementation, we might
wish to avoid realizing multipliers---which are expensive---altogether
and directly seek a permutation calculable with a small number of logic
gates. To this end, we extend an approach to LUT-free Hamming decoding 
provided in \cite{bliss}.

Firstly, we interpret the column index set $[2^{r-1}]$ 
as the vector space $\mathbb{F}_2^{r-1}$
taking $j \in [2^{r-1}]$ with binary representation 
$j = \sum_{u=0}^{r-2} b_u 2^u$ with $b_u \in \{ 0, 1 \}$
to correspond to the binary
vector $(b_0, b_1, \ldots, b_{r-2}) \in \mathbb{F}_2^{r-1}$.

Consider the class of permutations of $\mathbb{F}_2^{r-1}$
which map
$(b_0,b_1,\dots,b_{r-2})\in \mathbb{F}_2^{r-1}$ to 
$(c_0,c_1,\dots,c_{r-2})\in \mathbb{F}_2^{r-1}$
and take the form
\begin{align*}
	b_0 &\mapsto c_0 = b_0\\
	b_1 &\mapsto c_1 = b_1 \\
	b_i &\mapsto c_i = b_i + \prod_{u=0}^{d^{(i)}-1} (s^{(i)}_u + b^{(i)}_u)
	\;\text{for $2 \leq i \leq r-2$}
\end{align*}
where $d^{(i)}$ and 
$(s^{(i)}_0,s^{(i)}_1,\dots,s^{(i)}_{d^{(i)}-1})$ for $2 \leq i \leq r-2$
are parameters defining the permutation. We further
impose that $d^{(i)} \leq i$ which guarantees invertibility
since the expression
\begin{equation*}
	c_i = b_i + \prod_{u=0}^{d^{(i)}-1} (s^{(i)}_u + b^{(i)}_u)
\end{equation*} 
allows the recovery of $b_i$ from $c_i$ and $b_0,b_1,\dots,b_{i-1}$
and thus all $b_i$ from $c_i$ by induction. We will represent $(s^{(i)}_0,s^{(i)}_1,\dots,s^{(i)}_{d^{(i)}-1})$ 
with the corresponding integer $\sum_{u=0}^{d^{(i)}-1} s^{(i)}_u2^u$
for convenience.

This class of permutations generalizes the specific example
given in \cite{bliss} for $r = 8$ and is easy to implement in hardware
requiring for each $i$ a $d^{(i)}$-input AND 
gate, one XOR gate, and a NOT gate for every nonzero
value of $s^{(i)}_u$.

Table \ref{boolean-parameters} defines a special choice
of systematizing permutations for
extended Hamming codes of length up to $2^{16}$ 
with the property that the table can be truncated
at any row while maintaining the systematizing property.
These were found by a computer search.
For example, for extended Hamming codes of length
$2^6 = 64$, the map defined by Table \ref{boolean-parameters}
is
\[ 
	(b_0,b_1,b_2,b_3,b_4,b_5) \mapsto\\
	(b_0,b_1,b_2+b_0 b_1,b_3+b_0 b_1, b_4+ b_0 b_1 b_2,b_5+(1+b_0)b_1 b_2)
\]
The last seven columns of our parity-check matrix are
are the image of $\{57,58,\dots,63\}$ under this map evaluated over
$\mathbb{F}_2^6$ with an all-ones row appended. This yields
\begin{equation*}
	\begin{pmatrix}
		1  &1  &1  &1  &1  &0  &1\\
		1  &1  &1  &1  &1  &1  &0\\
		1  &1  &0  &1  &1  &1  &0\\
		0  &0  &1  &1  &1  &1  &0\\
		0  &1  &1  &0  &0  &1  &1\\
		1  &0  &1  &0  &1  &0  &1\\
		1  &1  &1  &1  &1  &1  &1
	\end{pmatrix}
\end{equation*}
which is invertible. Example code is also provided online
\cite{hosc-software}.

\section{Higher-Order Staircase Codes ($L\geq1$, $M\geq 1$)} \label{Generalized-TDZC}

Tiled diagonal zipper codes \cite{sukmadji-cwit} are a 
family of scattering generalizations of staircase
codes for which the encoding and decoding memory approaches
half of that of classical staircase codes while
maintaining a variable node degree of $M+1 = 2$. 
In this section, we further generalize our 
scattering generalized staircase codes 
from Section \ref{gsc-section} to both include 
this family and to extend it 
to arbitrary bit or variable node degree $M+1$. 
This generalization has two uses: The
first is that the encoding
and decoding memory can be reduced
while keeping all else essentially constant. 
The second is that it provides flexibility in hardware
designs by allowing the sidelength to be tuned while
keeping all else essentially constant.
The sidelength is a proxy for the 
degree of parallelism
that is naturally admitted by an encoding 
or decoding circuit. We refer to the resulting
scattering code 
family as \emph{higher-order staircase codes}
since it subsumes and extends the generalized
staircase codes of Definition \ref{gsc-definition}.

\subsection{Code Structure}

Consider starting with a weak generalized
staircase code of sidelength 
$S'$ and degree parameter $M'$.
We further have our permutations of 
$[S']\times [S']$ given by 
$\pi'_k$ for $k\in [M'+1]$
and the corresponding matrix-permuting functions
given by $\Pi'_k$ for $k\in [M'+1]$.
Finally, we have a ruler $(d_0,d_1,\dots,d_{M'})$
and a component code $\mathcal{C}$ of length
$(M'+1)S'$ and dimension $(M'+1)S'-r$. 
However, we have the following twist: We
define our constraint that the rows of
\begin{equation}\label{constraint-span-LZ}
	\big(
	\Pi'_{M'}(B_{n-d_{M'}}) \big\vert
	\cdots \big\vert
	\Pi'_{1}(B_{n-d_1}) \big\vert
	\Pi'_{0}(B_{n-d_0})
	\big)
\end{equation}
belong 
to $\mathcal{C}$ 
as holding for all integers $n\in L\mathbb{Z}$ where 
$L$ is a new fixed positive integer parameter. In other words, we 
restrict the constraint on the rows of \eqref{constraint-span-LZ}
to integers $n$ that are multiples of $L$.

Revisiting the Tanner graph perspective,
the constraint nodes become 
$\mathcal{L} = L\mathbb{Z} \times [S']$
and from \eqref{bit-neighbours}
we see that a variable node $(n,i,j)\in \mathcal{P}$
has the constraint $(n',i')\in\mathcal{L}$ as 
a neighbour if and only if $n' = n + d_k$ for
some $k \in [M'+1]$ where $n' \in L\mathbb{Z}$.
Therefore, we get a variable node degree of
\begin{equation}\label{irregular-degree}
	\hspace{-1ex}
	\abs{\mathcal{N}((n,i,j))} = 
	\abs{\{d_k \,\vert\, k \in [M'\hspace*{-0.75ex}+\hspace{-0.5ex}1], d_k \equiv -n \pmod*{L}\}}
\end{equation}
which depends on the block index $n$ and the ruler.
We seek to recover our original regular variable degree of 
$M+1$ for all variable nodes independently of $n$. 
This requires an \emph{$L$-uniform} ruler defined as follows:

\begin{definition}[$L$-uniform ruler]\label{uniform-ruler-def}
	A ruler $(d_0,d_1,\dots,d_{M'})$ of order $M'+1$ is 
	\emph{$L$-uniform} if it has an equal number of
	marks belonging to each of the residue classes 
	\begin{equation*}
		L\mathbb{Z}, L\mathbb{Z}+1, \dots, L\mathbb{Z}+L-1\text{.}
	\end{equation*}
	Necessarily, any such ruler can be constructed from a set
	of $L$ \emph{base} rulers $\mathcal{X} = \{X_0,X_1,\dots,X_{L-1}\}$ 
	each of order $M+1 = (M'+1)/L$ as the union of the
	entries of 
	\begin{equation*}
		LX_0, LX_1+1, \dots, LX_{L-1}+L-1
	\end{equation*}
	together with some ordering.
\end{definition}

\begin{remark}
	Normalizing an $L$-uniform ruler produces a ruler which is
	still $L$-uniform.
\end{remark}

\begin{remark}\label{uniform-ruler-partitioning-remark}
	Given $L$ order $M+1$ (base) rulers 
	\begin{equation*}
		X_\ell = (d_0^{(\ell)},d_1^{(\ell)},\dots,d_M^{(\ell)})
	\end{equation*}
	for $\ell\in [L]$, the elements of the corresponding 
	$L$-uniform ruler of order $L(M+1)$ can be arranged
	into an $L \times (M+1)$ array
	\begin{equation*}
		\begin{array}{llll}
			Ld_0^{(0)} & Ld_1^{(0)} & \cdots & Ld_M^{(0)} \\
			Ld_0^{(1)}+1 & Ld_1^{(1)}+1 & \cdots & Ld_M^{(1)}+1 \\
			\phantom{L}\vdots & \phantom{L}\vdots & \ddots & \phantom{L}\vdots\\
			Ld_0^{(L-1)}\!+\!L\!-\!1 & Ld_1^{(L-1)}\!+\!L\!-\!1& \cdots & Ld_M^{(L-1)}\!+\!L\!-\!1\\
		\end{array}
	\end{equation*}
	where the rows partition it into parts that are congruent modulo L
	and the columns partition it into parts that are not congruent modulo 
	L.
\end{remark}
From \eqref{irregular-degree} and 
Definition \ref{uniform-ruler-def},
we see that an $L$-uniform ruler
of order $L(M+1)$ leads to 
\begin{equation*}
	\abs{\mathcal{N}((n,i,j))} = M+1
\end{equation*}
for all $(n,i,j)\in \mathcal{P}$ 
as needed. Therefore, we can 
henceforth assume an $L$-uniform ruler as in 
Remark \ref{uniform-ruler-partitioning-remark}. 
We assume 
without loss of generality that
both the base rulers and the
corresponding $L$-uniform ruler
are normalized. Before
describing the encoding process
for the family of codes
obtained so far, we derive
the scattering conditions
which further reveal a
simplification:
Only $M+1$ distinct permutations 
are needed, one for each of the
columns of the array in Remark 
\ref{uniform-ruler-partitioning-remark}.

\subsection{Scattering Conditions}\label{hosc-scattering}

Observe that the inter-block
scattering condition
on \eqref{inter-block-condition}
becomes that for fixed $n'_1,n'_2\in
L\mathbb{Z}$ with $n'_2 > n'_1$,
\begin{equation*}
	n'_2 - n'_1 = d_{k'_2} - d_{k'_1} 
\end{equation*}
should have a  unique solution $k'_1,k'_2\in [M'+1]$ 
if any 
which necessarily satisfies 
$k_2' > k'_1$ \emph{and} 
$d_{k'_1}\equiv d_{k'_2} \pmod*{L}$
since 
$n'_2 - n'_1 \equiv 0 \pmod*{L}$.
We must then have for 
some $\ell \in [L]$ that
$d_{k'_1},d_{k'_2}\in \{Ld^{(\ell)}_k + \ell \mid k \in [M+1]\}$ thus
\begin{equation}\label{inter-block-condition-LZ}
	n'_2 - n'_1 = d_{k'_2} - d_{k'_1} = 
	L\cdot (d^{(\ell)}_{k_2} - d^{(\ell)}_{k_1})
\end{equation}
for some $\ell \in [L]$ and $k_1,k_2\in [M+1]$
with $k_2 > k_1$. Therefore, inter-block scattering
is obtained when \eqref{inter-block-condition-LZ}
has a unique solution for 
$\ell \in [L]$ and $k_1,k_2\in [M+1]$
with $k_2 > k_1$.

Secondly, the intra-block
scattering
condition on \eqref{intra-block-condition} becomes that 
\begin{equation}\label{intra-block-condition-LZ}
	\pi'_{k'_1}(i'_1,j'_1) = \pi'_{k'_2}(i'_2,j'_2) 
\end{equation}
must have a unique solution 
for $j'_1, j'_2 \in [S']$.
However, 
the uniqueness of 
solutions to \eqref{intra-block-condition-LZ}
need only be checked for $k'_1,k'_2\in [M'+1]$
that solve \eqref{inter-block-condition-LZ}
and thus correspond to $d_{k'_1}\equiv d_{k'_2} \pmod*{L}$.
We can then have $\pi'_{k'_1} = \pi'_{k'_2}$ for 
all $k'_1,k'_2\in [M'+1]$ such that 
$d_{k'_1} \not\equiv d_{k'_2} \pmod*{L}$.
This means that we can repeat every permutation
$L$ times so that only $(M'+1)/L = M+1$ distinct
permutations are needed to get intra-block scattering. 
In particular, consider a collection of $M+1$
permutations of
$[S']\times [S']$ given by 
$\pi_k$ for $k\in [M+1]$
and the corresponding matrix-permuting functions
given by $\Pi_k$ for $k\in [M+1]$. We can take
$\pi'_{k'} = \pi_{k}$ for 
every $k\in [M+1]$ and 
$k' \in [M'+1]$ such that 
$d_{k'} \in \{Ld^{(\ell)}_k + \ell \mid \ell \in [L]\}$.
Therefore, given $k'_1, k'_2 \in [M'+1]$,
$k_1,k_2\in[M+1]$, and $\ell \in L$ so that 
 $d_{k'_1} = Ld^{(\ell)}_{k_1} + \ell$ 
and 
$d_{k'_2} = Ld^{(\ell)}_{k_2} + \ell$, 
we have $\pi'_{k'_1} = \pi_{k_1}$ and $\pi'_{k'_2}=\pi_{k_2}$ 
so that \eqref{intra-block-condition-LZ} becomes to
\begin{equation}\label{intra-block-condition-LZ-reduced}
	\pi_{k_1}(i'_1,j'_1) = \pi_{k_2}(i'_2,j'_2)\text{.} 
\end{equation}
Intra-block scattering then corresponds to
uniqueness of solutions to \eqref{intra-block-condition-LZ-reduced} 
for $j'_1,j'_2 \in [S]$ where $k_1,k_2 \in [M+1]$. This
is obtained by taking $\pi_k$ for $k \in [M+1]$ to correspond
to an $(M+1,S')$-net.

\subsection{Encoding}

Note
that if the base rulers are
normalized, we get that $d_0^{(\ell)} = 0$ for each
$\ell \in [L]$. 
Observe then from Remark \ref{uniform-ruler-partitioning-remark}
that choosing the natural ordering for the corresponding
$L$-uniform ruler so that $d_0 < d_1 < \cdots < d_{L(M+1)-1}$
not only results in a normalized ruler, but also results in 
\begin{equation*}
	(d_0, d_1, \dots, d_{L-1}) = (0,1,\dots, L-1)\text{.}
\end{equation*}
This leads to
\begin{equation*}
	\pi'_0 = \pi'_1 = \cdots = \pi'_{L-1} = \pi_0
\end{equation*}
where $\pi_0$ is taken to be the identity
permutation. As a result, the 
constraint span \eqref{constraint-span-LZ} becomes
\begin{equation*}
	\big(
	\Pi'_{M'}(B_{n-d_{M'}}) \big\vert
	\hspace*{-0.5ex} \cdots \hspace*{-0.5ex} \big\vert
	\Pi'_{L}(B_{n-d_L})\big\vert
	B_{n-(L-1)} \big\vert
	\hspace*{-0.5ex} \cdots \hspace*{-0.5ex} \big\vert
	B_{n-1} \big\vert
	B_{n} 
	\big)
\end{equation*}
where $n \in L\mathbb{Z}$. Encoding is
performed $L$ blocks 
at a time with the $(n/L)$th $S' \times (LS'-r)$
block of information bits $B_{\frac{n}{L}}^{\mathsf{info}}$
adjoined to 
$M'-L+1$ permuted previous blocks 
to form
\begin{equation*}
	\big(
	\Pi'_{M'}(B_{n-d_{M'}}) \big\vert
	\hspace*{-0.5ex} \cdots \hspace*{-0.5ex} \big\vert
	\Pi'_{L+1}(B_{n-d_{L+1}})\big\vert
	\Pi'_{L}(B_{n-d_L})\big\vert
	B_{\frac{n}{L}}^{\mathsf{info}}
	\big)
\end{equation*}
whose rows are encoded
with systematic encoders 
for the component code 
$\mathcal{C}$ which is 
of dimension $(M'+1)S'-r$
and length $(M'+1)S'$.
This produces an $S'\times r$
by block of parity bits 
$B_{\frac{n}{L}}^{\mathsf{parity}}$
which is adjoined to the
information block to
produce $L$ consecutive
blocks as
\begin{equation*}
	\big(B_{n-(L-1)} \big\vert
	\hspace*{-0.5ex} \cdots \hspace*{-0.5ex} \big\vert
	B_{n-1} \big\vert
	B_{n} 
	\big)= 
	\big(
	\underbrace{\text{---}\; B_{\frac{n}{L}}^{\mathsf{info}}\; \text{---}}_{S' \times (LS'-r)}
	\;\big\vert\;
	\underbrace{\text{--}\; B_{\frac{n}{L}}^{\mathsf{parity}}\; \text{--}}_{S' \times r}
	\big)
\end{equation*}
such that the rows of $\eqref{constraint-span-LZ}$
belong to $\mathcal{C}$. 
As before, 
we can practically assume that
all blocks prior to the first
transmitted block are zero-valued.

To maintain the same code rate and component 
code $\mathcal{C}$ independently of $L$, we take
\begin{equation*}
	S' = \frac{S}{L}
\end{equation*}
so that the component code length
is $S'(M'+1) = S(M+1)$ and the 
code rate is 
\begin{equation*}
	R_\mathsf{unterminated} = 1-\frac{r}{LS'} = 1 - \frac{r}{S}\text{.}
\end{equation*}

\subsection{Higher-Order Staircase Codes}

It remains to define difference triangle sets (DTSs):

\begin{definition}[$(L,M)$-DTS {\cite[Part~VI:~Ch.~19]{Colbourn}}]
	An $(L,M)$-DTS is set of $L$ Golomb rulers
	$\mathcal{X} = \{X_0,X_1,\dots,X_{L-1}\}$ each of order
	$M+1$ whose respective sets of distances are disjoint. 
	Equivalently, $L$ order $M+1$ rulers given by 
	\begin{equation*}
		X_\ell = (d_0^{(\ell)},d_1^{(\ell)},\dots,d_M^{(\ell)})
	\end{equation*}
	where $\ell\in [L]$ form an $(L,M)$-DTS
	if all differences 
	$d^{(\ell)}_{k_1} - d^{(\ell)}_{k_2}$ 
	for $k_1,k_2\in [M+1]$ with $k_1\neq k_2$ and $\ell \in [L]$,
	are distinct. 
\end{definition}

It then follows that an $L$-uniform ruler of 
order $M'+1 = L(M+1)$ constructed by taking the 
base rulers to be an $(L,M)$-DTS 
leads to unique solutions to \eqref{inter-block-condition-LZ}
thus inter-block scattering. The construction is summarized
by the following definition:

\begin{definition}[Higher-order staircase code]\label{higher-order-staircase-code-def}
	Given all three of the following: 
	\begin{itemize}
		\item an $(L,M)$-DTS given by 
		\begin{equation*}
			0 = d^{(\ell)}_0 < d^{(\ell)}_1 < \dots < d^{(\ell)}_{M}
		\end{equation*}
		for $\ell\in [L]$ with corresponding $L$-uniform ruler of order $L(M+1)$
		\begin{equation*}
			d_0 < d_1 < \cdots < d_{L(M+1)-1}
		\end{equation*}
		given accordingly as
\[
			\{d_k \mid k \in [L(M+1)]\} = \\
			\{Ld^{(\ell)}_k + \ell \mid k \in [M+1], \ell \in [L]\}\text{;}
\]

		\item an $(M+1,S/L)$-net with corresponding  
		$M+1$ permutations of $[S/L]\times [S/L]$ given by $\pi_k$
		for $k \in [M+1]$ where $\pi_0$ is the identity permutation,
		and a resulting collection of 
		$L(M+1)$ permutations of $[S/L]\times [S/L]$ 
		given by $\pi'_{k'} = \pi_{k}$ for 
		every $k\in [M+1]$ and 
		$k' \in [L(M+1)]$ such that 
		$d_{k'} \in \{Ld^{(\ell)}_k + \ell \mid \ell \in [L]\}$\text{; and}
		\item a component code $\mathcal{C}$ of length $(M+1)S$
		and dimension $(M+1)S-r$\text{,}
	\end{itemize}
	a \emph{higher-order staircase code} of 
	rate $1-r/S$ is defined
	by the constraint on the bi-infinite sequence of 
	$(S/L) \times (S/L)$ matrices 
	$\dots,B_{-2},B_{-1},B_{0},B_{1},B_{2},\dots$ that the rows of 
	\begin{equation*}
		\big(
		\Pi'_{L(M+1)-1}(B_{n-d_{L(M+1)-1}}) \big\vert
		\cdots \big\vert
		\Pi'_{1}(B_{n-d_1}) \big\vert
		\Pi'_{0}(B_{n-d_0})
		\big)
	\end{equation*}
	belong to $\mathcal{C}$ for all $n \in L\mathbb{Z}$.
\end{definition}

Importantly, Section \ref{hosc-scattering} implies the following:
\begin{proposition}
	Higher-order staircase codes are scattering.
\end{proposition}

We provide a concrete example:
\begin{example}\label{hosc-example}
	Consider taking $L=2$ and $M = 2$. The
	$L = 2$ Golomb 
	rulers of order $M+1 = 3$ given by
	\begin{align*}
		(d_0^{(0)}, d_1^{(0)},d_2^{(0)}) &= (0,6,7)\\
		(d_0^{(1)}, d_1^{(1)},d_2^{(1)}) &= (0,2,5)
	\end{align*}
	form a $(2,2)$-DTS since their respective 
	sets of distances $\{1,6,7\}$ and 
	$\{2,3,5\}$ are disjoint. We further have
	\begin{align*}
		2(0,6,7) &= (0,12,14)\\
		2(0,2,5)+1 &= (1,5,11)
	\end{align*}
	thus a corresponding $2$-uniform ruler
	given by 
	\begin{equation*}
		(d_0,d_1,d_2,d_3,d_4,d_6) = (0,1,5,11,12,14)\text{.}
	\end{equation*}
	Finally, we have 
	\begin{align*}
		 (2d_0^{(0)}, 2d_0^{(1)}+1) &= (0,1) = (d_0,d_1)\\
		 (2d_1^{(0)}, 2d_1^{(1)}+1) &= (12,5) =  (d_4,d_2)\\
		 (2d_2^{(0)}, 2d_2^{(1)}+1) &= (14,11) = (d_5,d_3)
	\end{align*}
	thus 
	\begin{align*}
		\pi'_0 &= \pi'_1 = \pi_0\\
		\pi'_4 &= \pi'_2 = \pi_1\\
		\pi'_5 &= \pi'_3 = \pi_2\text{.}
	\end{align*}
	We consider the $(3,S/2)$-net permutations 
	of Example 
	\ref{net-perms} which are given 
	by $\pi_0(i,j)=(i,j)$, $\pi_1(i,j)=(j,i)$,
	and $\pi_2(i,j) = (j,i+j) \pmod{S/2}$. 
	This leads to 
	$\Pi_0(B) = B$, $\Pi_1(B) = B^\mathsf{T}$, 
	and we further define $B^\pi = \Pi_2(B)$.
	Our higher-order staircase code is then defined by
	the constraint that the rows of 
	\begin{equation*}
		\big(B_{n-14}^\pi \big\vert B_{n-12}^\mathsf{T} \big\vert B_{n-11}^\pi 
		\big\vert B_{n-5}^\mathsf{T}\big\vert B_{n-1}\big\vert B_n\big)
	\end{equation*}
	belong to some specified component code $\mathcal{C}$
	for all $n \in 2\mathbb{Z}$. 
\end{example}

\begin{figure}[t]
	\centering
	\includegraphics[]{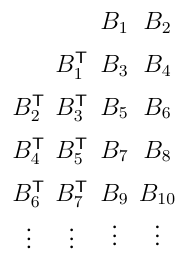}
	\caption{Visualization of a tiled diagonal zipper code
		corresponding to $L = 2$, $M = 1$, $(d_1^{(0)},d_1^{(1)}) = (1,2)$,
		and $\Pi_1(B)=B^\mathsf{T}$;
		rows belong to $\mathcal{C}$.}\label{tdzc-example-fig}
\end{figure}

\begin{figure}[t]
	\centering
	\includegraphics[]{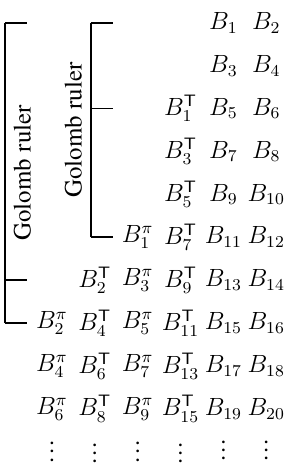}
	\caption{Visualization of the code of Example \ref{hosc-example}
		where $L = 2$ and $M = 2$; the respective sets of distances
		of the two Golomb rulers are disjoint and
		rows belong to $\mathcal{C}$.}\label{hosc-example-fig}
\end{figure}

A visualization of an $(L=2, M=1)$ code
corresponding to an instance of tiled
diagonal zipper codes \cite{sukmadji-cwit}
is provided in Fig.~\ref{tdzc-example-fig}.
A visualization of the $(L=2,M=2)$ code 
of Example \ref{hosc-example} follows it 
in Fig.~\ref{hosc-example-fig} illustrating
an increase in variable degree while
maintaining scattering.
In both of these pictures, the $L = 2$ rightmost columns 
represent the blocks that are actually transmitted 
while the columns to their left are delayed and permuted copies 
that are used to constrain the transmitted blocks.

Finally, we note that for higher-order staircase codes
as just defined, 
the code rate, $4$-cycle-freeness, component code, 
and variable and constraint node degrees are all independent of $L$.
It is therefore natural to study the effect of increasing 
$L$ while keeping all else fixed. 
We will find that while increasing $L$ reduces the block size to
$1/L^2$ times its size, the number of blocks that must be kept in memory
is $\Omega(L^2)$ meaning that we can at most reduce the memory
by a constant factor. However, this factor can still be meaningful
in practice and we may nonetheless need to minimize memory
for any given choice of $L$. Sections \ref{DTS-section} and \ref{computer-search-technique}
will deal with exhibiting DTSs that lead to
memory-optimal realizations of higher-order staircase codes.

Before proceeding, we comment briefly 
on the practical code performance implications of 
increasing $L$. Since the code rate, degree distribution, $4$-cycle-freeness, 
and component code are all unaffected by $L$, coding folklore 
suggests that the threshold and error floor 
performance should be \emph{roughly} independent of $L$.
For example, 
in the case of tiled diagonal zipper
codes corresponding to $(L \geq 1, M = 1)$, it
is found in \cite{sukmadji} that the 
gap to the hard-decision Shannon limit 
is increased by no more than $0.04$ dB as $L$ is increased
from $L = 1$ to $L = S = 1000$ with the component
code being a shortened triple-error-correcting
BCH code of length $2000$ and dimension $1967$.
A similar invariance can be observed for general higher-order 
staircase codes by experimenting with our simulators for these codes available
available online \cite{decsim-code,hosc-software}.
The reader is also referred to \cite{hbd-analysis}
for recent work on theoretical performance 
analysis under similar code families and decoding methods
to what is considered in this paper.
Moreover, the bound of Proposition \ref{stall-pattern-bound}
is independent of $L$ providing the same guarantee
on the coarse asymptotic behaviour of the error floor.

\subsection{Multiply-Chained and Other Extensions}

We now adapt a technique of Zhao et al.\ \cite{MC-TDZC}
in which multiple tiled diagonal zipper codes 
\cite{sukmadji-cwit} are chained in a circle to
boost their performance. This is straightforwardly
applicable to higher-order staircase codes. We introduce
a new positive integer parameter $C > 1$ and consider 
$C$ higher-order staircase code encoders indexed
by $c\in [C]$ and acting on $C$ independent input streams
with the following twist: The $c$th encoder uses the
memory of the $(c-1 \pmod*{C})$th encoder. This leads
to \emph{multiply-chained higher-order staircase codes}
formally defined as follows:

\begin{definition}[Multiply-chained higher-order staircase code]\label{multiply-chained-higher-order-staircase-code-def}
	Given all four of the following: 
	\begin{itemize}
		\item an integer $C > 1$;
		\item an $(L,M)$-DTS given by 
		\begin{equation*}
			0 = d^{(\ell)}_0 < d^{(\ell)}_1 < \dots < d^{(\ell)}_{M}
		\end{equation*}
		for $\ell\in [L]$ with corresponding $L$-uniform ruler of order $L(M+1)$
		\begin{equation*}
			d_0 < d_1 < \cdots < d_{L(M+1)-1}
		\end{equation*}
		given accordingly as
			\[
			\{d_k \mid k \in [L(M+1)]\} = \\
			\{Ld^{(\ell)}_k + \ell \mid k \in [M+1], \ell \in [L]\}\text{;}
			\]
		
		\item an $(M+1,S/L)$-net with corresponding  
		$M+1$ permutations of $[S/L]\times [S/L]$ given by $\pi_k$
		for $k \in [M+1]$ where $\pi_0$ is the identity permutation,
		and a resulting collection of 
		$L(M+1)$ permutations of $[S/L]\times [S/L]$ 
		given by $\pi'_{k'} = \pi_{k}$ for 
		every $k\in [M+1]$ and 
		$k' \in [L(M+1)]$ such that 
		$d_{k'} \in \{Ld^{(\ell)}_k + \ell \mid \ell \in [L]\}$\text{; and}
		\item a component code $\mathcal{C}$ of length $(M+1)S$
		and dimension $(M+1)S-r$\text{,}
	\end{itemize}
	a \emph{multiply-chained higher-order staircase code} of 
	rate $1-r/S$ is defined
	by the constraint on the $C$ bi-infinite sequences of 
	$(S/L) \times (S/L)$ matrices 
	$\dots,B_{-2}^{(c)},B_{-1}^{(c)},B_{0}^{(c)},B_{1}^{(c)},B_{2}^{(c)},\dots$ indexed by $c\in [C]$ that the rows of
	\[ 
		\bigg(
		\Pi'_{L(M+1)-1}(B_{n-d_{L(M+1)-1}}^{(c-1\pmod*{C})}) \bigg\vert\\ 
		\cdots \bigg\vert
		\Pi'_{L}(B_{n-d_{L}}^{(c-1\pmod*{C})}) \bigg\vert
		\Pi'_{L-1}(B_{n-d_{L-1}}^{(c)}) \bigg\vert
		\cdots
		\bigg\vert
		\Pi'_{0}(B_{n-d_0}^{(c)})
		\bigg)
	\]
	for each $c\in [C]$ belong to $\mathcal{C}$ for all $n \in L\mathbb{Z}$.
\end{definition}

We define a multiply-chained higher-order staircase code
with $C = 1$ to simply be a higher-order staircase code as in 
Definition \ref{higher-order-staircase-code-def} so that we
have a unified family defined for all $C\geq 1$. 
Lastly, we include one further generalization 
of this family which is as follows: Rather than considering
one component code $\mathcal{C}$, we can consider a ``time-varying'' family of component codes such as $\{\mathcal{C}_i^{(c)} \mid c \in [C],
i\in [S/L]\}$ where the $i$th row in the
constraint span in Definition \ref{multiply-chained-higher-order-staircase-code-def} belongs to $\mathcal{C}_i^{(c)}$. For example, this family can simply consist of $S/L$ permuted copies of a fixed code $\mathcal{C}$. With
this extension, the family of Definition \ref{multiply-chained-higher-order-staircase-code-def} includes 
the OFEC code family of \cite{OFEC} as the special case of 
$C = 2$, $M = 1$, and a special choice of time-varying
component code permutation defined therein.

\subsection{Memory}

We now consider the encoding memory requirements
for higher-order staircase codes. For the multiply-chained
variation, all requirements are simply multiplied by $C$ so
we restrict our attention to the case of $C = 1$.
At first glance,
encoding requires
that memory of the last $d_{M'}$
transmitted blocks be maintained.
A naive choice would be to maintain the
set of $L \cdot \max_{\ell\in [L]}d_M^{(\ell)}$ 
previous blocks $\mathsf{BUFF}_{\frac{n}{L}}$ 
defined by
\[
	\mathsf{BUFF}_{\frac{n}{L}}
	= \bigg\{B_{n-\mu} \;\bigg\vert\; \\ \mu \in
	 \left\{L,L+1,\dots, L \cdot \max_{\ell\in [L]}d_M^{(\ell)} + L-1\right\}
	 \bigg\}
\]
for encoding the $(n/L)$th information
block. However, we can observe in 
the case of Example \ref{hosc-example} illustrated
in Fig.~\ref{hosc-example-fig} that blocks
with even indices must be recalled for longer 
than blocks with odd indices. In general,
by examining the ruler decomposition 
given in Remark \ref{uniform-ruler-partitioning-remark},
we see that we can alternatively 
maintain $L$ independent buffers 
$\mathsf{BUFF}^{(\ell)}_{\frac{n}{L}}$ 
of sizes 
$d_M^{(\ell)}$ for $\ell \in [L]$ given by
\begin{equation*}
	\mathsf{BUFF}^{(\ell)}_{\frac{n}{L}}
	= \{B_{n-\mu} \mid \mu = L\nu + \ell, 
	\nu\in \{1,2,\dots,d_M^{(\ell)}\}\}
\end{equation*}
so that the total number of blocks stored is 
\begin{equation*}
	\sum_{\ell=0}^{L-1} d_M^{(\ell)}
	\leq L \cdot \max_{\ell\in [L]} d_M^{(\ell)}\text{.}
\end{equation*}
We then have an encoding memory of 
\begin{equation}\label{encoding-memory}
	\left(\frac{S}{L}\right)^2 \cdot \sum_{\ell=0}^{L-1} d_M^{(\ell)}
\end{equation}
bits (or nonbinary symbols). This is analogous 
to the concept of \emph{total memory} or 
\emph{overall constraint length} in the theory 
of convolutional codes 
\cite[Ch.~2]{Fund-Conv-Coding} with the key differences
being that bits are replaced with 
$(S/L)\times (S/L)$ blocks which admit permutation operations
and that one-bit memory elements are replaced 
accordingly with stores of $(S/L)^2$ bits. The
encoder style assumed here is analogous to
Forney's \emph{obvious realization} \cite{Forney-Convolutional}.

On the other hand, the decoding window must typically
be chosen as a number of blocks which is a constant 
multiple of $d_{L(M+1)-1}+1$, the smallest number of consecutive blocks 
which contains a complete component codeword. If 
we order our base rulers in descending order of length 
so that 
\begin{equation*}
	\max_{\ell\in [L]} d_M^{(\ell)} = d_M^{(0)}
	> d_M^{(1)} > \cdots > d_M^{(L-1)}
\end{equation*}
thus $d_{L(M+1)-1} = L\cdot d_M^{(0)} = L\cdot \max_{\ell\in [L]} d_M^{(\ell)}$,
we get
\begin{equation}\label{decoding-memory}
	\left(\frac{S}{L}\right)^2 \cdot \left(1+L\cdot\max_{\ell\in [L]} d_M^{(\ell)}\right)
\end{equation}
as a measure of the decoding memory. This
is analogous to the \emph{constraint length}
in the theory of convolutional codes.

We emphasize that \eqref{encoding-memory}
and \eqref{decoding-memory} are merely \emph{models}
for the encoding and decoding memory 
under a certain encoder
and decoder structure considered in \cite{sukmadji-cwit} 
which can
be subverted in practice. There are two main ways 
in which this is done: 
\begin{itemize}
	\item Each $(S/L)\times (S/L)$ block is compressed down 
	to a $(S/L)\times r$  future additive contribution to parity
	and thus is not directly stored.
	\item An encoder style analogous to the 
	\emph{observer canonical form} as opposed to the
	\emph{controller canonical form} \cite[Ch.~2]{Fund-Conv-Coding}
	or \emph{obvious realization} \cite{Forney-Convolutional}
	is used in which
	the encoding memory is proportional to  
	$\max_{\ell\in [L]} d_M^{(\ell)}$
	rather than $\sum_{\ell=0}^{L-1} d_M^{(\ell)}$.
	This encoder style was first recognized
	by Massey \cite{Massey}.
\end{itemize}
Both of these are done in the 
hardware implementations of a staircase code encoder
corresponding to
$(L=1, M=1)$ and a zipper code encoder 
corresponding to $(L = 32, M=1)$ 
found in \cite{Staircase-Hardware-Encoder}
and \cite{Zipper-Hardware-Encoder} respectively.
It is found in 
\cite{Zipper-Hardware-Encoder}
that there are significant hardware efficiency 
benefits to $L = 32$ relative to $L = 1$.
Nonetheless, one or both of
$\sum_{\ell=0}^{L-1} d_M^{(\ell)}$ and 
$\max_{\ell\in [L]} d_M^{(\ell)}$ 
are measures of hardware complexity.
Thus, we seek, for a given $L$ and $M$, 
an $(L,M)$-DTS for which $\sum_{\ell=0}^{L-1} d_M^{(\ell)}$  
and/or 
$\max_{\ell\in [L]} d_M^{(\ell)}$
are minimized. We defer the study of this problem
to Section \ref{DTS-section}.

\section{Extended-Hamming-Based Higher-Order Staircase Codes}\label{hosc-sim-section}

We consider now the use of memory-optimal 
higher-order staircase codes and their multiply-chained variation 
with extended Hamming component codes.
As discussed in the previous section, the construction of
such codes requires the construction
of DTSs with minimum scope and sum-of-lengths.
Such DTSs are provided in the Appendix with the subject
of how they are obtained deferred to 
Sections \ref{DTS-section} and \ref{computer-search-technique}.
We use the optimal DTSs from the Appendix together with the extended Hamming codes defined by Table \ref{systematic-Hamming-table} to construct our codes.

The encoder produces a $C(S/L) \times L (S/L)$ rectangle of coded bits 
since each of the $C$ higher-order staircase code encoders produces $L$ blocks
at a time each of dimensions $(S/L) \times (S/L)$. The decoding window size 
$W$ counts the number of such rectangles meaning that the number of bits
in the decoding window is $WC(S/L)^2L$. Similarly to before, we assume transmission 
across a BSC and
perform sliding window decoding of 
$W$ consecutive rectangles at a time  where one iteration 
comprises decoding all constraints in the window consecutively and $I$ iterations are performed
before advancing the window by one rectangle. We again use a termination-like procedure
as in Section \ref{instantiation-section} in which a frame is formed from 
$F$ consecutive rectangles. We take $F = 10^5 + W$ for all simulations
in this section so that the entire coding scheme is determined 
by the six-tuple of parameters $(L,M,S/L,C,W,I)$. The simulation curves
in Figures \ref{094-sims} and \ref{087-sims} are labeled with this
six-tuple appended with the resulting 
unterminated code rate $R_\mathsf{unterminated}$
and the decoding window size in bits for the reader's convenience.

\begin{figure}[t]
	\centering
	\includegraphics[width=0.5\columnwidth]{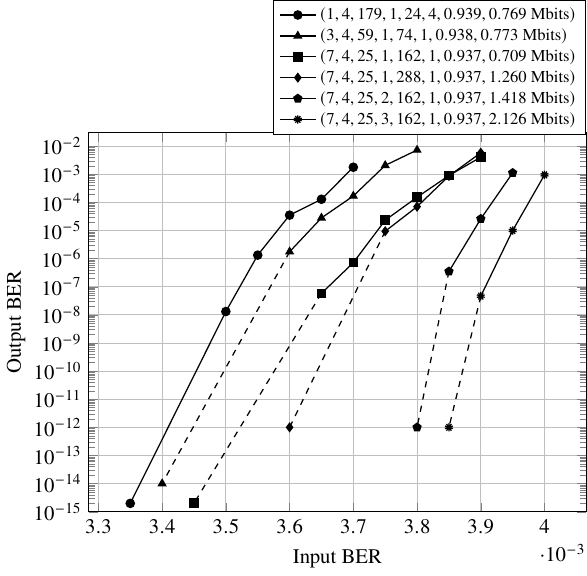}
	\caption{Simulation results for $R_\mathsf{unterminated} \approx 0.937$ examples with parameters $(L,M,S/L,C,W,I,R_\mathsf{unterminated},WC(S/L)^2L)$}\label{094-sims}
\end{figure}

\begin{figure}[t]
	\centering
	\includegraphics[width=0.5\columnwidth]{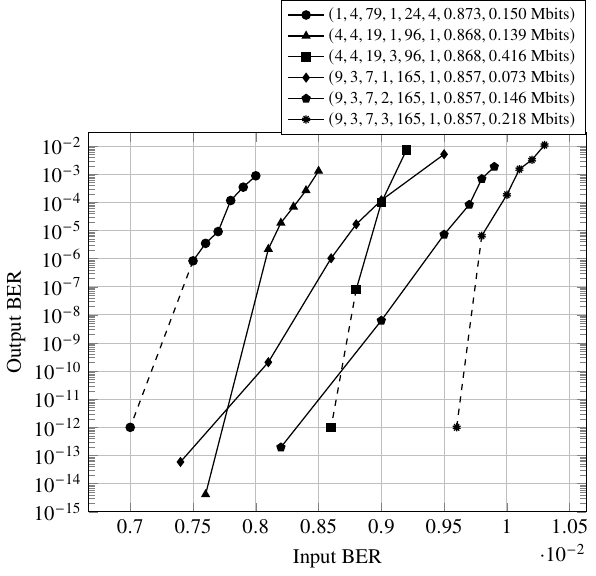}
	\caption{Simulation results for $R_\mathsf{unterminated} \approx 0.87$ examples with parameters $(L,M,S/L,C,W,I,R_\mathsf{unterminated},WC(S/L)^2L)$}\label{087-sims}
\end{figure}

One of our goals in this section is to demonstrate
the ability to construct a variety of codes with similar properties
by varying $L$ at a roughly constant code rate. 
We accordingly organize our simulations slightly differently with results for two different approximate
code rates given in Figures \ref{094-sims} and \ref{087-sims} respectively. 
The dashed lines indicate points at which
zero bit errors were measured after transmitting a number
of bits equal to ten times the reciprocal of the corresponding output BER.

Observe firstly that Fig.~\ref{087-sims} illustrates control over the error floor
and threshold via $M$ and $C$ respectively as predicted. 
Moreover, the examples in these figures perform
comparably to the codes compared in \cite{smith,MC-TDZC} while using only
Hamming components as opposed to BCH components and also improve 
over the results of Section \ref{instantiation-section} with generalized 
staircase codes $(L = 1)$ in terms of memory and decoding iterations. 

We provide heuristic estimates of the performance, latency,
and complexity for some examples as before. In Table \ref{complexity-comparison-table-094-hosc},
we compare the classical staircase code of \cite{smith} to the extended-Hamming-based 
higher-order staircase code with parameters
$(L,M,S/L,C,W,I) = (7,4,25,1,162,1)$.
As before, we assume that for the relatively small values of $M$
considered, the cost
of each component code decoding is dominated by the cost of invoking the 
component code decoder $t^2$ rather than the cost of updating
the syndromes of codewords involved in each corrected bit.
In Table \ref{complexity-comparison-table-087-hosc}, we compare
a multiply-chained tiled diagonal zipper code from \cite{MC-TDZC}
with parameters $(L,M,S/L,C,W,I) = (2,1,125,2,8,4)$ and 
BCH component codes with $t = 3$ to the extended-Hamming-based 
higher-order staircase code with parameters 
$(L,M,S/L,C,W,I) = (4,4,19,1,96,1)$. In both situations,
we see large reductions in latency and complexity at a
performance cost of roughly $0.3$ to $0.5$ dB.

\begin{table}[t]
	\centering \caption{Heuristic estimates of performance, latency, and complexity for higher-order staircase codes with $R_\mathsf{unterminated} \approx 0.937$}
	\begin{tabular}{r|c|c}
		& staircase \cite{smith,Dmitri-Hardware-Staircase} & this work \\\hline
		gap to Shannon at 1e-15 & 0.56 dB & 0.89 dB \\
		latency: $WC(S/L)^2L$ & 1.8e6 & 7.1e5 \\
		iterations: $I$ & 3 & 1 \\
		decodings per iteration: $WC(S/L)$ & 3.6e3 & 4.1e3 \\
		cost per decoding: $t^2$ & 9 & 1 \\
		complexity score: $IWC(S/L)t^2$ & 9.7e4 & 4.1e3
	\end{tabular} \label{complexity-comparison-table-094-hosc}
\end{table}
\begin{table}[t]
	\centering \caption{Heuristic estimates of performance, latency, and complexity for higher-order staircase codes with $R_\mathsf{unterminated} \approx 0.87$}
	\begin{tabular}{r|c|c}
				& Fig.\ 13 in \cite{MC-TDZC} & this work \\\hline
		gap to Shannon at 1e-8 & 0.75 dB & 1.25 dB \\
		latency: $WC(S/L)^2L$ & 5e5 & 1.4e5 \\
		iterations: $I$ & 4 & 1 \\
		decodings per iteration: $WC(S/L)$ & 2e3 & 1.8e3 \\
		cost per decoding: $t^2$ & 9 & 1 \\
		complexity score: $IWC(S/L)t^2$ & 7.2e4 & 1.8e3
	\end{tabular} \label{complexity-comparison-table-087-hosc}
\end{table}

We anticipate improved
performance--complexity--latency tradeoffs in a wider variety 
of contexts in light of recent works \cite{sukmadji-irregular,Huawei-Higher-Degree} 
which accomplish such goals via new product-like codes
with bit degree increased from two to three. The simulations
here are only meant to provide a few motivating examples. 
These results can be reproduced and extended 
using our simulation
tools available online \cite{decsim-code,hosc-software}.

\section{Difference Triangle Sets with Minimum Scope and Sum-of-Lengths}\label{DTS-section}

Consider
an $(L,M)$-DTS given by 
\begin{equation*}
	0 = d^{(\ell)}_0 < d^{(\ell)}_1 < \dots < d^{(\ell)}_{M}
\end{equation*}
for $\ell \in [L]$. As with Golomb rulers or $(1,M)$-DTSs, 
it is easy to construct an $(L,M)$-DTS for arbitrary $L$ and $M$
if no further objectives 
are placed. A well-studied
hard problem is that
of constructing an $(L,M)$-DTS 
of minimum \emph{scope}
where scope is the largest mark 
across all (normalized) rulers
$\max_{\ell\in [L]} d_M^{(\ell)}$.
Equivalently, scope is the maximum length among
all rulers. We again have a \emph{trivial bound} 
on the scope of an $(L,M)$-DTS given by
\begin{equation}\label{scope-lower-bound}
	\max_{\ell\in [L]} d_M^{(\ell)}
	\geq 
	L \cdot \frac{(M+1)M}{2}
\end{equation}
with equality if and only if 
the union of the sets of distances
of the $L$ rulers is precisely $\{1,2,\dots,L(M+1)M/2\}$.
An $(L,M)$-DTS for which the trivial bound 
holds with equality is said to be \emph{perfect}.
The union of the $L$ (disjoint) sets of distances
of the $L$ rulers forming an $(L,M)$-DTS will be referred 
to concisely as its \emph{set of distances}.

On the other hand, we also require
a minimization of the 
\emph{sum-of-lengths} 
$\sum_{\ell=0}^{L-1} d_M^{(\ell)}$.
To the
best of the authors' knowledge, this
problem has not been explicitly studied.
We begin by identifying when 
minimizing scope is equivalent
to minimizing sum-of-lengths.
Aside from the case of 
$L = 1$ where sum-of-lengths and scope are
trivially equivalent, we have the following:

\begin{proposition}\label{scope-slen-minimization-equivalence}
	A minimum scope $(L,M)$-DTS also has minimum
	sum-of-lengths if $M \in \{1,2\}$ and it is perfect. 
\end{proposition}
\begin{proof}
	In the case of $M = 2$, observe that
	\begin{align}
		2\sum_{\ell = 0}^{L-1} d_2^{(\ell)} &= 
		\sum_{\ell = 0}^{L-1}
		((d_2^{(\ell)}-d_1^{(\ell)}) 
		+ d_1^{(\ell)}
		+ d_2^{(\ell)})\label{sum-over-distances}\\
		&\geq \frac{3L(3L+1)}{2}\label{1to3L}\text{.}
	\end{align}
	In particular, \eqref{sum-over-distances} is
	a sum over all three distances of all $L$ rulers
	which must be distinct positive integers. This
	leads to the lower bound \eqref{1to3L} which is 
	achieved precisely when perfect.
	When $M = 1$, the sum of 
	the lengths is precisely the sum of the 
	distances which means that a similar
	argument holds. 
\end{proof}

Beyond this case, the problem of minimizing
sum-of-lengths will be distinct and non-trivial.
If $M > 2$, there is no simple 
proportionality relationship
such as \eqref{sum-over-distances}
between a sum over all ruler lengths and 
a sum over all distances. Moreover,
even in the case of $M = 2$, a minimum scope $(L,2)$-DTS
which is not perfect does not necessarily 
minimize sum-of-lengths as illustrated by 
the following counterexample:

\begin{example}\label{scope-sum-of-lengths-non-equivalence}
	Consider the $(2,2)$-DTS of Example \ref{hosc-example} given by
	$\{(0,2,5),(0,6,7)\}$ which is of scope $7$. Observe that
	this scope is already minimum since 
	achieving the scope lower bound of $6$ is impossible 
	since it requires equality of \eqref{sum-over-distances}
	which is even and \eqref{1to3L} which is $21$. 
	However, the alternative $(2,2)$-DTS given by 
	$\{(0,3,4),(0,2,7)\}$ 
	which also has minimum scope also has a smaller 
	sum-of-lengths. 
\end{example}

Except for the cases of $M = 1$
and $M = 2$, it remains an open problem to 
exhibit minimum scope $(L,M)$-DTSs for arbitrary 
$L$. However, perfect
DTSs admit recursive constructions which lead to
constructive existence results for 
infinitely many values of $L$. 
It is known that perfect DTSs
can only exist for $M \leq 4$. 
Therefore, for the cases of 
$M = 3$ and $M = 4$, recursive
constructions provide
the strongest existence results. 
The values of
$L$ covered by such constructions 
are often (but not always) too large to be of relevance to
applications. Therefore, such constructions
do not 
eliminate the need for computer-search 
constructions for smaller values of $L$
but are nonetheless at least of theoretical
interest.
The reader is referred to \cite{Colbourn-IT,Klove}
and the references therein for more on this 
context.
We will study the implications
of one such construction due to 
Wild \cite{Wild} and Kotzig and Turgeon \cite{Kotzig-Turgeon}
for the problem of minimizing sum-of-lengths. We will
need the following:

\begin{definition}[Sharply $2$-transitive permutation group]
	A \emph{sharply $2$-transitive permutation group
	$G$ acting on a set $X$} 
	is a \emph{group}
	of permutations of $X$
	with the following property: 
	For each pair of ordered pairs 
	$(w,x), (y,z) \in X\times X$ with $w\neq x$
	and $y\neq z$, there exists a \emph{unique}
	permutation $g\in G$ such that
	$(g(w),g(x)) = (y,z)$.
\end{definition}

As noted
in \cite{Wild}, we can exhibit a
finite such group
as the group of invertible
affine transformations from a finite
field $\mathbb{F}_{M+1}$ to itself $x \mapsto ax+b$
where $M+1$ is a prime power and 
$a,b \in \mathbb{F}_{M+1}$ with $a\neq 0$
yielding $M(M+1)$ such permutations. 
For $M = 3$, the resulting 
permutations of $[M+1]$
in array format are
\begin{equation*}
\begin{array}{cccccccccccc}
	0 & 1 & 2 & 3 & 0 & 1 & 2 & 3 & 0 & 1 & 2 & 3 \\
	1 & 0 & 3 & 2 & 2 & 3 & 0 & 1 & 3 & 2 & 1 & 0 \\
	2 & 3 & 0 & 1 & 3 & 2 & 1 & 0 & 1 & 0 & 3 & 2 \\
	3 & 2 & 1 & 0 & 1 & 0 & 3 & 2 & 2 & 3 & 0 & 1 \\
\end{array}
\end{equation*}	
and for $M = 4$, the permutations
are defined by simple modular
arithmetic since $\mathbb{F}_{5}$
is $\mathbb{Z}_5$.

\begin{theorem}[Special case of the combining construction of
	Wild \cite{Wild} and Kotzig and Turgeon \cite{Kotzig-Turgeon}]\label{combining-construction}
	Let $M+1$ be a prime power and let the collection of $M(M+1)$ permutations
	\begin{align*}
		\rho_k\colon [M+1] &\longrightarrow [M+1]\\
		i &\longmapsto \rho_k(i)
	\end{align*}
	where $k\in [M(M+1)]$ be 
	a sharply $2$-transitive permutation group acting on $[M+1]$.
	Let $\mathcal{X}$ be a perfect $(L_1,M)$-DTS and 
	$\mathcal{Y}$ 
	be a perfect $(L_2,M)$-DTS with each given respectively by 
	\begin{equation*}
		\mathcal{X} = \{(x_0^{(i)},x_1^{(i)},\dots,x_M^{(i)}) \mid i \in [L_1] \}
	\end{equation*}
	and
	\begin{equation*}
		\mathcal{Y} = \{(y_0^{(i)},y_1^{(i)},\dots,y_M^{(i)}) \mid i \in [L_2] \}\text{.}
	\end{equation*}
	Let $\mathcal{W}$ be the set of $L_1L_2M(M+1)$ rulers given by
\[
		(L_2M(M+1)+1)\left(x_0^{(i)},x_1^{(i)},\dots,x_M^{(i)}\right)
		+\\
		\left(y_{\rho_k(0)}^{(j)},y_{\rho_k(1)}^{(j)},\dots,y_{\rho_k(M)}^{(j)}\right)
\]
	for $i\in [L_1]$, $j\in [L_2]$, and $k\in [M(M+1)]$.
	Let $\mathcal{Z}$ be the union of $\mathcal{Y}$, $(L_2M(M+1)+1)\mathcal{X}$, and $\mathcal{W}$.
	Then, $\mathcal{Z}$ is a perfect $(L_1L_2M(M+1) + L_1 + L_2, M)$-DTS. (However, the rulers
	are in unnormalized form.)
\end{theorem}

\begin{proposition}[Sum-of-lengths under combining]\label{sum-of-lengths-under-combining}
	Combining a perfect $(L_1,M)$-DTS with sum-of-lengths $S_1$ and a perfect 
	$(L_2,M)$-DTS with sum-of-lengths $S_2$ as in Theorem \ref{combining-construction} yields a perfect $(L_1L_2M(M+1)+L_1+L_2, M)$-DTS with sum-of-lengths $(L_2M(M+1)+1)^2S_1+S_2$.
\end{proposition}

\begin{proof}
	Trivially, $\mathcal{Y}$ and $(L_2M(M+1)+1)\mathcal{X}$ contribute
	$S_2$ and $(L_2M(M+1)+1)S_1$ to the sum-of-lengths so it remains
	to calculate the contribution of $\mathcal{W}$. The difference
	between the $v$th and $u$th mark of the ruler in
	$\mathcal{W}$ corresponding to a fixed
	$i\in [L_1]$, $j\in [L_2]$, and $k\in [M(M+1)]$
	can be expressed as 
	\begin{equation*}
		(L_2M(M+1)+1)\left(x_v^{(i)}-x_u^{(i)}\right) + \left(y_{\rho_k(v)}^{(j)}-y_{\rho_k(u)}^{(j)}\right)\text{.}
	\end{equation*}
	Since the term $y_{\rho_k(v)}^{(j)}-y_{\rho_k(u)}^{(j)}$ satisfies
	\begin{equation*}
		-\frac{L_2M(M+1)}{2} \leq y_{\rho_k(v)}^{(j)}-y_{\rho_k(u)}^{(j)} \leq \frac{L_2M(M+1)}{2}\text{,}
	\end{equation*}
	the scale factor $(L_2M(M+1)+1)$ guarantees that the ordering
	of the differences $x_v^{(i)}-x_u^{(i)}$ is preserved. Assuming
	without loss of generality that the rulers of $\mathcal{X}$
	are normalized so that the length of the $i$th ruler of
	$\mathcal{X}$ is the difference corresponding to
	$v = M$ and $u = 0$, the same difference will then yield
	the length of the corresponding rulers in $\mathcal{W}$
	as
	\begin{equation}\label{ruler-length-in-W}
		(L_2M(M+1)+1)x_M^{(i)} + \left(y_{\rho_k(M)}^{(j)}-y_{\rho_k(0)}^{(j)}\right)\text{.}
	\end{equation}
	Lastly, sharp $2$-transitivity guarantees that
	\begin{equation*}
		\sum_{k = 0}^{M(M+1)-1}
		y_{\rho_k(M)}^{(j)}-y_{\rho_k(0)}^{(j)} = 0
	\end{equation*}
	in which case 
	summing \eqref{ruler-length-in-W}
	over all $i\in [L_1]$, $j\in [L_2]$, and $k\in [M(M+1)]$,
	yields 
	\begin{equation*}
		L_2M(M+1)(L_2M(M+1)+1)S_1
	\end{equation*}
	which is added to $S_2$ and $(L_2M(M+1)+1)S_1$
	to obtain the desired result.
\end{proof}

We will now consider each of the cases of
$M\in\{1,2,3,4\}$.
These cases are both practically 
the most relevant for applications
to higher-order staircase codes but
are also theoretically
interesting for the reasons
just discussed.

\subsection{The Case of $M = 1$}

For $(L,1)$-DTSs, 
we have the trivial bound 
\begin{equation*}
	\max_{\ell\in [L]} d_1^{(\ell)} \geq L
\end{equation*}
and trivially
\begin{equation}\label{deg-2-slen-lb}
	\sum_{\ell=0}^{L-1} d_1^{(\ell)} \geq \frac{L(L+1)}{2}
\end{equation}
with both of these bounds holding with equality
for the perfect $(L,1)$-DTS given by 
$\{(0,1),(0,2),\dots,(0,L)\}$.

\subsection{The Case of $M = 2$}

For $(L,2)$-DTSs, it is known that comparing the parities of 
\eqref{sum-over-distances} and \eqref{1to3L} 
slightly strengthens the trivial bound to
\begin{equation}\label{deg-3-scope-lb}
	\max_{\ell\in [L]} d_2^{(\ell)} \geq 
	\begin{cases}
		3L & \text{if } L \equiv 0 \text{ or } 1 \pmod*{4} \\
		3L + 1 & \text{if } L \equiv 2 \text{ or } 3\pmod*{4}
	\end{cases}\text{.}
\end{equation}
From this and \eqref{sum-over-distances}, it straightforwardly
follows that
\begin{equation}
\label{deg-3-slen-lb}
	\sum_{\ell=0}^{L-1} d_2^{(\ell)}\\ \geq  
	\begin{dcases}
		\frac{3L(3L+1)}{4} & \text{if } L \equiv 0 \text{ or } 1 \pmod*{4} \\
		\frac{(3L-1)3L}{4} + \frac{3L+1}{2} & \text{if } L \equiv 2 \text{ or } 3\pmod*{4}
	\end{dcases}
\end{equation}
where the first inequality holds with equality if
and only if the set of distances is precisely $\{1,2,\dots,3L\}$
and the second inequality holds with equality if and only if
set of distances is precisely $\{1,2,\dots,3L-1,3L+1\}$.

An explicit construction of $(L,2)$-DTSs 
achieving the bound \eqref{deg-3-scope-lb}
exists for arbitrary $L$ and was
given by Skolem \cite{Skolem} in the case of 
$L \equiv 0 \text{ or } 1 \pmod*{4}$
and by O'Keefe \cite{OKeefe} in the case of
$L \equiv 2 \text{ or } 3\pmod*{4}$.
In the case of $L \equiv 0 \text{ or } 1 \pmod*{4}$,
by perfectness, the bound \eqref{deg-3-slen-lb} is 
automatically achieved. In the case of 
$L \equiv 2 \text{ or } 3 \pmod*{4}$,
the construction of \cite{OKeefe} happens
to preclude $3L$ from the set of distances
and hence also achieves \eqref{deg-3-slen-lb}.
Note that Example \ref{scope-sum-of-lengths-non-equivalence}
demonstrates how it is possible to achieve \eqref{deg-3-scope-lb}
without achieving \eqref{deg-3-slen-lb}.
Nonetheless, the specific constructions of 
\cite{Skolem} and \cite{OKeefe} simultaneously
achieve \eqref{deg-3-scope-lb} and \eqref{deg-3-slen-lb}
for arbitrary $L$ settling this case.
We provide the constructions of \cite{Skolem,OKeefe}
in the Appendix for the reader's convenience since
\cite{Skolem,OKeefe} consider a different but
equivalent problem to the $(L,2)$-DTS problem
requiring a well-known but slightly tedious translation.

\subsection{The Case of $M = 3$}

For $(L,3)$-DTSs, the trivial bound is
\begin{equation}\label{deg-4-scope-lb}
	\max_{\ell\in [L]} d_3^{(\ell)} \geq 6L
\end{equation}
and it is a longstanding conjecture 
of Bermond \cite{Bermond} that this is achievable
for all $L$ excluding $L\in \{2,3\}$. 
Theorem \ref{combining-construction}
implies that there are infinitely
many values of $L$ for which 
\eqref{deg-4-scope-lb} is achieved, by, 
for example, repeatedly
combining the perfect $(1,3)$-DTS
$\{(0,1,4,6)\}$ with itself.
Bermond's conjecture has been verified
for all $L \leq 1000$ in \cite{deg-4-1000}
by using a wide variety of recursive constructions
to cover gaps as needed. 

\begin{proposition}
	The sum-of-lengths of 
	an $(L,3)$-DTS
	satisfies 
	\begin{equation}\label{deg-4-slen-lb}
		\sum_{\ell = 0}^{L-1}d_3^{(\ell)} \geq 5L^2+L
		\text{.}
	\end{equation}
	\begin{proof}
		By the distinctness of positive differences, we have
		\begin{align*}
			\sum_{u=1}^{6L} u
			&\leq\sum_{\ell = 0}^{L-1}\sum_{i = 0}^{3}\sum_{j=i+1}^{3}
			d_j^{(\ell)}-d_i^{(\ell)}\\
			&= 
			4\sum_{\ell = 0}^{L-1} d_3^{(\ell)}-\sum_{\ell = 0}^{L-1}(d^{(\ell)}_1-d^{(\ell)}_0)+
			(d^{(\ell)}_3-d^{(\ell)}_2)\\
			&\leq 
			4\sum_{\ell = 0}^{L-1} d_3^{(\ell)}-\sum_{u=1}^{2L}u
		\end{align*}
		which yields the desired result after rearranging.
	\end{proof}
\end{proposition}

It is not the case that $(L,3)$-DTSs achieving 
\eqref{deg-4-scope-lb} such as those listed in 
\cite{Shearer-IBM-Optimal-List} achieve \eqref{deg-4-slen-lb}.
However, we conjecture that it is
possible to simultaneously achieve \eqref{deg-4-scope-lb}
and \eqref{deg-4-slen-lb} for all $L$ excluding
$L\in \{2,3,4,5\}$.
By using a computer-search technique described in
Section \ref{computer-search-technique}, we found
$(L,3)$-DTSs simultaneously achieving \eqref{deg-4-scope-lb}
and \eqref{deg-4-slen-lb} for all $6\leq L \leq 15$
and these are provided in the Appendix.
In the cases of $L\in \{2,3,4,5\}$, 
one or both of \eqref{deg-4-scope-lb}
and \eqref{deg-4-slen-lb} are not achievable
but simultaneous absolute minimization of
scope and sum-of-lengths is still possible
and the optimal solutions are provided in the
Appendix.
These cases are small enough to be
handled by an exact 
mixed-integer linear programming (MILP)
solver by slightly modifying the
MILP formulation of the minimum scope DTS 
problem given in \cite{MILP-DTS}. 
Alternatively,
an exhaustive backtracking DTS
search program of Shearer \cite{Shearer-IBM-Programs,Shearer-IBM-Programs-Paper} 
can be used for these small cases.

Next, we consider constructing
an infinite family via Theorem
\ref{combining-construction}.
For all integers $i \geq 0$, 
let $\mathcal{Z}_i$ be a 
perfect $(L_i,3)$-DTS
with sum-of-lengths 
$S_i$ obtained by combining 
$\mathcal{Z}_{i-1}$ and $\mathcal{Z}_{0}$
via Theorem \ref{combining-construction}
where $\mathcal{Z}_0$ is a fixed 
perfect $(L_0,3)$-DTS with sum-of-lengths
$S_0$. Proposition \ref{sum-of-lengths-under-combining}
provides linear recurrences for $L_i$ and $S_i$
which we can solve to get
\begin{equation*}
	L_i = \frac{1}{12}\left((12L_0+1)^{i+1}-1\right)
\end{equation*} 
and
\begin{equation*}
	S_i = 
	\frac{6S_0}{6L_0^2+L_0}L_i^2
	+ 
	\frac{S_0}{6L_0^2+L_0}L_i\text{.}
\end{equation*}
If $\mathcal{Z}_0$ satisfies
$\eqref{deg-4-slen-lb}$ so that
$S_0 = 5L_0^2+L_0$, we get
\begin{equation*}
	S_i = 
	\left(5+\frac{1}{6L_0+1}\right)L_i^2
	+\left(1-\frac{L_0}{6L_0+1}\right)L_i\text{.}
\end{equation*}
Since the largest $L_0$ for which we have exhibited
a perfect $(L_0,3)$-DTS achieving \eqref{deg-4-slen-lb}
in the Appendix is $L_0 = 15$, the best result we can get is
as follows:
\begin{proposition}\label{infinite-deg-4-fam}
	There are infinitely many values of $L$ for which a
	perfect $(L,3)$-DTS with sum-of-lengths
	\begin{equation*}
		\left(5+\frac{1}{91}\right)L^2+\left(1-\frac{15}{91}\right)L
	\end{equation*}
	exists.
\end{proposition}

\subsection{The Case of $M = 4$}
For $(L,4)$-DTSs, it is known that the
trivial bound can be slightly strengthened
by parity considerations to
\begin{equation}\label{deg-5-scope-lb}
	\max_{\ell\in [L]} d_4^{(\ell)} \geq 
	\begin{cases}
		10L & \text{if } L \equiv 0 \pmod*{2} \\
		10L + 1 & \text{if } L \equiv 1 \pmod*{2}
	\end{cases}\text{.}
\end{equation}

\begin{proposition}
	The sum-of-lengths of an $(L,4)$-DTS
	satisfies
	\begin{equation}\label{deg-5-slen-lb}
		\sum_{\ell=0}^{L-1} d_4^{(\ell)}\\ \geq  
		\begin{dcases}
			9L^2+\frac{3}{2}L & \text{if } L \equiv 0 \pmod*{2} \\
			9L^2+\frac{3}{2}L + \frac{1}{2} & \text{if } L \equiv 1\pmod*{2}
		\end{dcases}\text{.}
	\end{equation}
\end{proposition}
\begin{proof}
	By the distinctness of positive differences,
	we have
\[
		2\sum_{\ell = 0}^{L-1} d_4^{(\ell)}
		= 
		\sum_{\ell = 0}^{L-1}\bigg[ 
		(d^{(\ell)}_2-d^{(\ell)}_0) +
		(d^{(\ell)}_4-d^{(\ell)}_2)+
		(d^{(\ell)}_1-d^{(\ell)}_0) +\\
		(d^{(\ell)}_2-d^{(\ell)}_1) +
		(d^{(\ell)}_3-d^{(\ell)}_2) +
		(d^{(\ell)}_4-d^{(\ell)}_3) 
		\bigg]
		\geq 
		\sum_{u=1}^{6L} u
\]
	which yields the result after integrality
	considerations.
\end{proof}

We conjecture that both \eqref{deg-5-scope-lb}
and \eqref{deg-5-slen-lb} are simultaneously achievable
for all $L$ excluding $L\in \{2,3,4\}$.
We provide in the Appendix DTSs
simultaneously achieving \eqref{deg-5-scope-lb} and \eqref{deg-5-slen-lb}
for all $5 \leq L \leq 10$ excluding $L = 9$
for which a search is ongoing.  
The problem of simultaneously achieving 
\eqref{deg-5-scope-lb} and \eqref{deg-5-slen-lb}
was considered implicitly by Laufer who solved it in the case of $L \in \{6,8\}$ in \cite{Laufer}. 
This was done
because the structure implied
by achieving 
\eqref{deg-5-slen-lb} can be exploited 
to simplify the enumeration of solutions.
A solution achieving both of \eqref{deg-5-scope-lb}
and \eqref{deg-5-slen-lb} for $L = 10$ was 
found by Khansefid et al.\ \cite{KTG,Shearer-IBM-Optimal-List,Shearer-IBM-Optimal-List-Credits}.
In the cases of $L\in \{5,7\}$, we found
solutions
with Shearer's search programs \cite{Shearer-IBM-Programs,Shearer-IBM-Programs-Paper}.
For $L\in \{2,3,4\}$, it is not
only the case that one or both of \eqref{deg-5-scope-lb} and \eqref{deg-5-slen-lb}
is not achievable, but also that we get a tradeoff between
minimizing scope and minimizing sum-of-lengths when $L = 4$. In
this case, we provide a pair of Pareto optimal solutions. These smaller
cases again can be handled by MILP or Shearer's search programs \cite{Shearer-IBM-Programs,Shearer-IBM-Programs-Paper}.

For all integers $i \geq 0$, 
let $\mathcal{Z}_i$ be a 
perfect $(L_i,4)$-DTS
with sum-of-lengths 
$S_i$ obtained by combining 
$\mathcal{Z}_{i-1}$ and $\mathcal{Z}_{0}$
via Theorem \ref{combining-construction}
where $\mathcal{Z}_0$ is a fixed 
perfect $(L_0,4)$-DTS with sum-of-lengths
$S_0$. Solving the 
linear recurrences for $L_i$ and $S_i$, we get
\begin{equation*}
	L_i = \frac{1}{20}\left((20L_0+1)^{i+1}-1\right)
\end{equation*} 
and
\begin{equation*}
	S_i = 
	\frac{10S_0}{10L_0^2+L_0}L_i^2
	+ 
	\frac{S_0}{10L_0^2+L_0}L_i\text{.}
\end{equation*}
We must have even $L_0$ for $\mathcal{Z}_0$ to be perfect. Therefore,
taking $\mathcal{Z}_0$ satisfying $\eqref{deg-5-slen-lb}$ yields
$S_0 = 9L_0^2+(3/2)L_0$ thus
\begin{equation*}
	S_i = 
	\left(9+\frac{6}{10L_0+1}\right)L_i^2
	+\left(\frac{3}{2}-\frac{6L_0}{10L_0+1}\right)L_i\text{.}
\end{equation*}
Taking $\mathcal{Z}_0$ to be the perfect $(10,4)$-DTS achieving \eqref{deg-5-slen-lb} in the Appendix, we get the following:
\begin{proposition}\label{infinite-deg-5-fam}
	There are infinitely many values of $L$ for which a
	perfect $(L,4)$-DTS with sum-of-lengths
	\begin{equation*}
		\left(9+\frac{6}{101}\right)L^2+\left(\frac{3}{2}-\frac{60}{101}\right)L
	\end{equation*}
	exists.
\end{proposition}

\subsection{Implications for Higher-Order Staircase Codes}

We now consider some implications
of the results of this section 
for higher-order staircase codes:

\subsubsection{Recovery of Tiled Diagonal Zipper Codes}

Consider the higher-order staircase code defined by
the $(2,S/L)$-nets from 
Example \ref{net-perms} or Example \ref{net-perms-invo}
and the scope and sum-of-lengths optimal $(L,1)$-DTS
given by 
$d_1^{(\ell)} = L-\ell$ for $\ell \in [L]$.
We get that
\begin{align*}
	\{d_L,d_{L+1},\dots,d_{2L-1}\}
	&= \{L(L-\ell) + \ell \mid \ell \in [L]\}\\
	&= \{L(\ell+2) -(\ell+1) \mid \ell \in [L]\}
\end{align*}
which yields 
\begin{equation*}
	\pi'_{L} = \pi'_{L+1} = \dots = 
	\pi'_{2L-1} = \pi_1
\end{equation*}
where $\pi_1(i,j) = (j,i)$.
Our higher-order staircase code 
is then defined by the constraint
that the rows of
\begin{equation*}
	\big(
	B_{n-L^2}^\mathsf{T}
	\big\vert
	\hspace*{-0.5ex} \cdots \hspace*{-0.5ex}
	\big\vert
	B_{n-(3L-2)}^\mathsf{T} \big\vert
	B_{n-(2L-1)}^\mathsf{T} \big\vert
	B_{n-(L-1)} \big\vert
	\hspace*{-0.5ex} \cdots \hspace*{-0.5ex} \big\vert
	B_{n-1} \big\vert
	B_{n} 
	\big)
\end{equation*}
belong to some component code $\mathcal{C}$
for all $n\in L\mathbb{Z}$.
This family is precisely the 
tiled diagonal zipper
codes of \cite{sukmadji-cwit}.

\subsubsection{Memory Reduction Limits}

Tiled diagonal zipper
codes
are often colloquially 
described as providing encoding
and decoding memory reductions
by up to $1/2$ relative to classical
staircase codes \cite{smith} 
which are the $L=1$ case for
this family.
The factor of $1/2$ arises as a 
certain reasonable limit
representing the approximate
memory reduction at high code
rates. We can
compute analogous limits for 
higher-order
staircase codes. In particular, we divide the encoding
memory \eqref{encoding-memory} by the
encoding memory in the $L=1$ case
and take the limit as $S\to \infty$
with $S/L$ held constant 
so that $L\to \infty$
in each of the cases of $M\in \{1,2,3,4\}$.
We will assume for each $M\in\{1,2,3,4\}$
the existence of infinitely many values
of $L$ for which the respective sum-of-lengths
lower bounds \eqref{deg-2-slen-lb}, \eqref{deg-3-slen-lb},
\eqref{deg-4-slen-lb}, and \eqref{deg-5-slen-lb} are achieved.
The resulting limits are $1/2$, $3/4$, $5/6$, and $9/11$
respectively
for each $M\in\{1,2,3,4\}$.
Note that our premise is only conjectured to be
true in the cases of $M\in\{3,4\}$ but Propositions
\ref{infinite-deg-4-fam} and \ref{infinite-deg-5-fam}
allow us to prove limits of $5/6 + 1/546$ and 
$9/11 + 6/1111$ in those cases instead. 
We can also obtain decoding memory reduction
limits via \eqref{decoding-memory} similarly
as $1/2$, $3/4$, $6/7$, and $10/12$ respectively
for each $M\in\{1,2,3,4\}$.

\subsubsection{Recovery of Self-Orthogonal Convolutional Codes}
Consider a higher-order staircase code with $S = L$
and recall Definition \ref{trivial-net-family}
for $(M+1,1)$-nets. We see that the net permutations disappear
from the construction. Consider further taking $r = 1$, 
i.e., assuming single-parity-check constraints. 
The definition of higher-order staircase codes (Definition \ref{higher-order-staircase-code-def})
still remains
interesting in this case and leads to a 
family of rate $(L-1)/L$ codes. These 
codes are precisely a recursive rather than
feedforward version of the self-orthogonal
convolutional codes of Robinson and Bernstein \cite{CSOC}.
In particular, the codes of \cite{CSOC} 
have a slightly increased rate of $L/(L+1)$
arising from an $(L,M)$-DTS 
achieved by omitting 
previous parities in the recursive encoding 
process. The reader is referred to \cite{FF-Staircase}
for an example of how the conversion to a feedforward
variant can be done in the case of classical 
staircase codes. Such conversions typically lead 
to large error floors since parity symbols are
left unprotected.

\section{DTS Search Techniques} \label{computer-search-technique}

A highly-optimized variation on the DTS search algorithm
of Koubi et al.\ \cite{Koubi-Hill-Climbing} was implemented.
Part of this variation is based on an adaptation of a 
standard programming trick in which a subset of $[N]$ is represented
by a single $N$-bit integer. A critical, innermost loop
step of the algorithm of Koubi et al.\ \cite{Koubi-Hill-Climbing}
and any similar stochastic local search algorithm
is the insertion of a new mark into a partially 
constructed DTS. This requires that the distances
created by the new mark are compared to all of
the distances used by the partial DTS thus
far, requiring up to $\mathcal{O}(LM^3)$ 
operations. Under a C language
model \cite{C-language}, this
can be reduced to an essentially constant number
of operations as will be seen shortly.

The resulting algorithm can also be used
to make progress on the classical
minimum scope DTS problem. As a certificate of
its effectiveness, we exhibit a perfect (thus minimum scope) $(12,4)$-DTS and a minimum scope $(13,4)$-DTS
settling the minimum scope DTS problems for these cases and improving on the previous best
upper bounds of $122$ and $133$ given in \cite{Shearer-IBM-Report}:
\begin{align*}
	\begin{array}{rrrrr}
		0 & 3 & 62 & 106 & 120 \\
		0 & 11 & 66 & 86 & 119 \\
		0 & 27 & 34 & 105 & 118 \\
		0 & 18 & 56 & 99 & 116 \\
		0 & 22 & 51 & 74 & 115 \\
		0 & 10 & 42 & 77 & 114 \\
		0 & 6 & 63 & 89 & 113 \\
		0 & 2 & 47 & 87 & 112 \\
		0 & 19 & 80 & 95 & 111 \\
		0 & 21 & 70 & 100 & 109 \\
		0 & 12 & 48 & 94 & 102 \\
		0 & 28 & 96 & 97 & 101 \\
	\end{array}
	\quad
	\begin{array}{rrrrr}
		0 & 44 & 80 & 115 & 131 \\
		0 & 33 & 42 & 123 & 130 \\
		0 & 1 & 69 & 109 & 129 \\
		0 & 25 & 73 & 84 & 127 \\
		0 & 31 & 41 & 96 & 126 \\
		0 & 26 & 64 & 78 & 125 \\
		0 & 21 & 66 & 112 & 124 \\
		0 & 8 & 83 & 100 & 122 \\
		0 & 23 & 93 & 117 & 121 \\
		0 & 15 & 82 & 101 & 119 \\
		0 & 29 & 56 & 105 & 118 \\
		0 & 6 & 63 & 113 & 116 \\
		0 & 32 & 34 & 106 & 111 \\
	\end{array}
\end{align*}

\subsection{C Language Model Speedup}

We review key elements of the $C$ language model 
\cite{C-language} to which the reader is
referred for further details as needed:
\begin{itemize}
	\item An $N$-bit unsigned integer $\mathsf{n}$ is
	interpreted as an element of $\mathbb{Z}_{2^N}$ and 
	its $i$th bit is the value
	of $\mathsf{b}_i \in \{0,1\}$ where $\mathsf{n} = \sum_{i=0}^{N-1} \mathsf{b}_i2^i$.
	\item The binary operators $\BitwiseAND$, $\BitwiseOR$, $\ShiftLeft$, and $\ShiftRight$
	respectively denote bitwise logical AND, bitwise logical OR, logical left shift, and logical
	right shift.
	\item All non-zero integers are interpreted as logically true and zero is logically false. The binary operators $\LogicalAND$ and $\LogicalOR$ and the unary operator $\LogicalNOT$ respectively denote logical AND, logical OR, and logical NOT acting accordingly on integer operands.
\end{itemize}

In what follows, \emph{all variables} are $N$-bit unsigned integers
which have a dual interpretation.
When a variable is treated as a set, it
represents the subset of $[N]$ whose elements are all $i \in [N]$ such that the $i$th bit of the variable is $1$.
Otherwise, a variable simply represents its integer value.
A \emph{partially
	constructed DTS} is a set of rulers with differing
numbers of elements such that all differences between 
elements within a ruler are distinct from each other
and distinct from differences produced by other rulers.
We are given a value $\mathsf{mark}$ which we wish to
insert into a particular ruler $\mathsf{natRuler}$ 
of a partially constructed DTS, but only if it does not
lead to repeated differences. We maintain:
\begin{itemize}
	\item the set of differences used by the partially constructed DTS thus
	far $\mathsf{usedDistances}$,
	\item the largest mark of $\mathsf{natRuler}$ given by $\mathsf{largestMark}$, and
	\item a redundant mirror image of $\mathsf{natRuler}$ given by $\mathsf{revRuler}$
	where $\mathsf{revRuler}$ contains $i$ if and only if $\mathsf{natRuler}$ contains 
	$\mathsf{largestMark}-i$.
\end{itemize}
If $\mathsf{mark}$ is successfully inserted into $\mathsf{natRuler}$, these
three variables must also be updated accordingly.

\begin{algorithm}[t]
	\caption{Fast conditional mark insertion.}\label{fast-conditional-mark-insertion}
	\begin{algorithmic}[1]
		\If{$\mathsf{mark} > \mathsf{largestMark}$} 
		\State $\mathsf{leftDistances} \gets \mathsf{revRuler} \ShiftLeft (\mathsf{mark}-\mathsf{largestMark})$
		\State $\mathsf{rightDistances} \gets 0$
		\Else
		\State $\mathsf{leftDistances} \gets \mathsf{revRuler} \ShiftRight (\mathsf{largestMark}-\mathsf{mark})$
		\State $\mathsf{rightDistances} \gets \mathsf{natRuler} \ShiftRight \mathsf{mark}$
		\EndIf
		\State $\mathsf{bisection} \gets \mathsf{leftDistances} \BitwiseAND \mathsf{rightDistances}$
		\State $\mathsf{distances} \gets \mathsf{leftDistances} \BitwiseOR \mathsf{rightDistances}$
		\State $\mathsf{intersection} \gets \mathsf{usedDistances} \BitwiseAND \mathsf{distances}$
		\If{$\LogicalNOT(\mathsf{bisection} \LogicalOR \mathsf{intersection})$}
		\State $\mathsf{natRuler} \gets \mathsf{natRuler} \BitwiseOR (1\ShiftLeft \mathsf{mark})$ 
		\If{$\mathsf{mark} > \mathsf{largestMark}$}
		\State $\mathsf{revRuler}\gets \mathsf{leftDistances} \BitwiseOR 1$
		\State $\mathsf{largestMark}\gets\mathsf{mark}$
		\Else
		\State $\mathsf{revRuler} \gets \mathsf{revRuler} \BitwiseOR (1\ShiftLeft (\mathsf{largestMark}-\mathsf{mark}))$ 
		\EndIf
		\State $\mathsf{usedDistances}\gets \mathsf{usedDistances} \BitwiseOR \mathsf{distances}$
		\EndIf
	\end{algorithmic}
\end{algorithm}

We accomplish all of this with Algorithm \ref{fast-conditional-mark-insertion}.
Here, $\mathsf{leftDistances}$ is assigned the distances between
$\mathsf{mark}$ and all existing marks that are less than 
or equal to $\mathsf{mark}$. On the other hand, 
$\mathsf{rightDistances}$ is assigned the distances between
$\mathsf{mark}$ and all existing marks that are greater than 
or equal to $\mathsf{mark}$. It can be shown that 
these two sets intersect
if and only if the inserted $\mathsf{mark}$ is equal to an 
existing mark or bisects a pair of existing marks which leads
to a repeated distance. Otherwise, the only possible repeated
distances are common distances between $\mathsf{distances}$ and $\mathsf{usedDistances}$.

Finally, we must contend with the issue that in practice,
we have $N=64$ in accordance 
with the $64$-bit data widths of modern processors. 
This apparently limits us
to searching for DTSs with scope at most $64$. However, we 
can emulate a larger $N$ by implementing bitwise and shifting operations 
sequentially on $64$-bit chunks. This leads to increased
complexity relative to what is presented
as Algorithm \ref{fast-conditional-mark-insertion}, 
especially for the shifting operations which require
carryovers. This also means that
the number of operations is not strictly 
constant. However, there nonetheless remains 
a significant overall speedup relative to an 
array-based implementation performing mark-by-mark 
comparisons. 
An emulated $256$-bit implementation of
Algorithm \ref{fast-conditional-mark-insertion}
can be found within our
DTS search program sample implementation,
available online \cite{mohannad-dts-search-github}.

\subsection{Mark Statistics}

The algorithm of Koubi et al.\ \cite{Koubi-Hill-Climbing}
involves a simple stochastic local search procedure in
which a DTS is populated with randomly generated marks 
which are inserted if compatible as just described.
Backtracking steps such as deleting an entire ruler 
or a mark within a ruler are used when no compatible
marks are found after a specified number of attempts.
These are accommodated by adding appropriate
record-keeping procedures to Algorithm \ref{fast-conditional-mark-insertion}.
However, the key novel aspect of the algorithm
of Koubi et al.\ \cite{Koubi-Hill-Climbing}
is the modelling of the marks of a ruler in a DTS 
as arising from a simple Gaussian mixture: 
For each $i\in [M+1]$, 
the $i$th mark of any ruler 
in the DTS is sampled independently 
from a Gaussian distribution 
with mean $\mu_i$ and variance $\sigma^2_i$ (with rounding
to the nearest integer).
These parameters can be inferred in various 
ways described in \cite{Koubi-Hill-Climbing}
and a followup work \cite{Koubi-Followup}.

We employ the following procedure 
for obtaining the mark distribution parameters
$\mu_i$ and $\sigma^2_i$: 
Suppose that we seek DTSs with
scope at most $T$. We run a version 
of our algorithm with a 
uniform mark distribution to find some number of DTSs
with a larger admissible scope $\widetilde T$ with $\widetilde T > T$. 
If the resulting sample means and variances are
$\widetilde\mu_i$ and $\widetilde\sigma^2_i$,
we take $\mu_i = (T/\widetilde{T})\widetilde{\mu}_i$
and $\sigma_i^2 = (T/\widetilde{T})^2\widetilde{\sigma}_i^2$.
If we seek DTSs which further 
minimize sum-of-lengths, we 
add a third essential step:
We generate a number of DTSs with
scope at most $T$ and take $\mu_i$
and $\sigma^2_i$ to be
the means and variances of a biased
sample consisting only of those
DTSs whose sum-of-lengths is below
a specified threshold, e.g., below average.
This results in an algorithm which is significantly
more likely to produce a minimum sum-of-lengths
DTS.

Aside from the key elements just described, 
there are some other minor differences
from the algorithm of \cite{Koubi-Hill-Climbing} 
which we will not bother describing;
instead, we provide a sample implementation 
online \cite{mohannad-dts-search-github}.

\section{Future Work}\label{Conclusion}

Several directions for future work are immediately
apparent. On the application side, 
while we have provided an essentially
five-parameter 
($L$, $M$, $S/L$, $C$, and $t$) family of codes, we
have only evaluated performance for a few parameters 
with $t = 1$ and with hard-decision decoding.
Many other possibilities
can be considered in future work such as 
$t > 1$, soft-decision decoding, and any techniques
that have been applied to product-like codes.
Such investigations can be initiated by
using and extending our simulators available online \cite{decsim-code,hosc-software}
but can also be aided by the development of threshold 
and error floor analysis techniques along the lines of 
\cite{sukmadji,Lei,hbd-analysis}. 
Furthermore, we may also consider 
the development of non-scattering and irregular-degree
variants as done in \cite{Huawei-Higher-Degree} and 
\cite{sukmadji-irregular} respectively.

On the mathematical side, it may be within reach 
to construct an 
infinite family of DTSs achieving the bounds
\eqref{deg-4-slen-lb} or \eqref{deg-5-slen-lb}
or to strengthen Propositions \ref{infinite-deg-4-fam}
and \ref{infinite-deg-5-fam} by considering the many other
construction methods for perfect DTSs provided in 
\cite{deg-4-1000}. It may also be within reach
to find an interesting general sum-of-lengths lower bound
for arbitrary $M$. Lastly, the search for 
sum-of-lengths-optimal and scope-optimal DTSs can be continued.

{\appendix[Explicit Difference Triangle Sets with Minimum Scope and Sum-of-Lengths]}

\subsection{$(L \geq 1, M=2)$}
For $1\leq L \leq 7$, we have:
\begin{align*}
	\begin{array}{rrr}
		0 & 1 & 3 \\
	\end{array}
	\quad
	\begin{array}{rrr}
		0 & 2 & 7 \\
		0 & 3 & 4 \\
	\end{array}
	\quad
	\begin{array}{rrr}
		0 & 3 & 10 \\
		0 & 6 & 8 \\
		0 & 4 & 5 \\
	\end{array}
	\quad
	\begin{array}{rrr}
		0 & 11 & 12 \\
		0 & 4 & 10 \\
		0 & 2 & 9 \\
		0 & 3 & 8 \\
	\end{array}\\[5pt]
	\begin{array}{rrr}
		0 & 13 & 15 \\
		0 & 6 & 14 \\
		0 & 11 & 12 \\
		0 & 3 & 10 \\
		0 & 5 & 9 \\
	\end{array}
	\quad
	\begin{array}{rrr}
		0 & 3 & 19 \\
		0 & 12 & 17 \\
		0 & 6 & 15 \\
		0 & 1 & 14 \\
		0 & 7 & 11 \\
		0 & 2 & 10 \\
	\end{array}
	\quad 
	\begin{array}{rrr}
		0 & 2 & 22 \\
		0 & 1 & 19 \\
		0 & 6 & 17 \\
		0 & 9 & 16 \\
		0 & 12 & 15 \\
		0 & 4 & 14 \\
		0 & 5 & 13 \\
	\end{array}
\end{align*}

For all $L \geq 8$, we have: 
\vspace{1pt}
\subsubsection{$L=4m$ with $m > 1$ due to Skolem \cite{Skolem}}
\vspace{-1pt}
\begin{align*}
	&(0,4m-1,10m)\\
	&(0,2m-1,8m-1)\\
	&(0,1,5m+1)\\
	&(0,4m-2i,12m-i) \text{ for } i\in\{0,1,\dots,2m-1\}\\
	&(0,4m-1-2i,8m-1-i) \text{ for } i\in\{1,2,\dots,m-1\}\\
	&(0,2m-3-2i,7m-1-i) \text{ for } i\in\{0,1,\dots,m-3\}
\end{align*}
\subsubsection{$L=4m+1$ with $m > 1$ due to Skolem \cite{Skolem}}
\vspace{-1pt}
\begin{align*}
	&(0,4m+1,10m+3)\\
	&(0,2m-1,8m+2)\\
	&(0,1,5m+3)\\
	&(0,4m-2i,12m+3-i) \text{ for } i\in\{0,1,\dots,2m-1\}\\
	&(0,4m+1-2i,8m+2-i) \text{ for } i\in\{1,2,\dots,m\}\\
	&(0,2m-1-2i,7m+2-i) \text{ for } i\in\{1,2,\dots,m-2\}
\end{align*}
\subsubsection{$L=4m+2$ with $m > 1$ due to O'Keefe \cite{OKeefe}}
\vspace{-1pt}
\begin{align*}
	&(0,4m+1,10m+4)\\
	&(0,2m+1,10m+5)\\
	&(0,4m+2,12m+7)\\
	&(0,1,11m+6)\\
	&(0,4m+2-2i,8m+4-i) \text{ for } i\in\{1,2,\dots,2m\}\\
	&(0,4m+1-2i,12m+6-i) \text{ for } i\in\{1,2,\dots,m-1\}\\
	&(0,2m+1-2i,11m+5-i) \text{ for } i\in\{1,2,\dots,m-1\}
\end{align*}
\subsubsection{$L=4m+3$ with $m > 1$ due to O'Keefe \cite{OKeefe}}
\vspace{-1pt}
\begin{align*}
	&(0,2m+3,7m+6)\\
	&(0,1,5m+5)\\
	&(0,2m+1,8m+6)\\
	&(0,4m+2,10m+8)\\
	&(0,4m+3,12m+10)\\
	&(0,4m+2-2i,12m+9-i) \text{ for } i\in\{1,2,\dots,2m\}\\
	&(0,4m+3-2i,8m+6-i) \text{ for } i\in\{1,2,\dots,m-1\}\\
	&(0,2m+1-2i,7m+6-i) \text{ for } i\in\{1,2,\dots,m-1\}
\end{align*}

\subsection{$(1 \leq L \leq 15, M=3)$}

These are found with the methods described in Section \ref{computer-search-technique}.

\begin{align*}
	\begin{array}{rrrr}
		0 & 1 & 4 & 6 \\
	\end{array}
	\quad
	\begin{array}{rrrr}
		0 & 3 & 12 & 13 \\
		0 & 5 & 7 & 11 \\
	\end{array}
	\quad
	\begin{array}{rrrr}
		0 & 3 & 15 & 19 \\
		0 & 1 & 10 & 18 \\
		0 & 2 & 7 & 13 \\
	\end{array}\\[5pt]
	\begin{array}{rrrr}
		0 & 6 & 14 & 24 \\
		0 & 3 & 22 & 23 \\
		0 & 9 & 16 & 21 \\
		0 & 4 & 15 & 17 \\
	\end{array}
	\quad
	\begin{array}{rrrr}
		0 & 8 & 21 & 30 \\
		0 & 14 & 26 & 29 \\
		0 & 4 & 23 & 28 \\
		0 & 7 & 25 & 27 \\
		0 & 1 & 11 & 17 \\
	\end{array}
	\quad
	\begin{array}{rrrr}
		0 & 10 & 32 & 36 \\
		0 & 8 & 24 & 35 \\
		0 & 5 & 33 & 34 \\
		0 & 12 & 25 & 31 \\
		0 & 7 & 21 & 30 \\
		0 & 2 & 17 & 20 \\
	\end{array}\\[5pt]
	\begin{array}{rrrr}
		0 & 14 & 30 & 42 \\
		0 & 1 & 38 & 41 \\
		0 & 8 & 32 & 39 \\
		0 & 10 & 27 & 36 \\
		0 & 13 & 33 & 35 \\
		0 & 5 & 23 & 34 \\
		0 & 4 & 19 & 25 \\
	\end{array}
	\quad
	\begin{array}{rrrr}
		0 & 2 & 43 & 48 \\
		0 & 13 & 36 & 47 \\
		0 & 15 & 33 & 45 \\
		0 & 6 & 28 & 44 \\
		0 & 7 & 39 & 42 \\
		0 & 9 & 26 & 40 \\
		0 & 8 & 27 & 37 \\
		0 & 4 & 24 & 25 \\
	\end{array}
	\quad
	\begin{array}{rrrr}
		0 & 13 & 53 & 54 \\
		0 & 4 & 49 & 52 \\
		0 & 12 & 34 & 51 \\
		0 & 14 & 43 & 50 \\
		0 & 15 & 38 & 47 \\
		0 & 16 & 35 & 46 \\
		0 & 2 & 26 & 44 \\
		0 & 6 & 27 & 37 \\
		0 & 5 & 25 & 33 \\
	\end{array}\\[5pt]
	\begin{array}{rrrr}
		0 & 19 & 42 & 60 \\
		0 & 15 & 43 & 59 \\
		0 & 10 & 57 & 58 \\
		0 & 17 & 49 & 56 \\
		0 & 3 & 53 & 55 \\
		0 & 20 & 46 & 54 \\
		0 & 13 & 40 & 51 \\
		0 & 9 & 31 & 45 \\
		0 & 12 & 33 & 37 \\
		0 & 5 & 29 & 35 \\
	\end{array}
	\quad
	\begin{array}{rrrr}
		0 & 10 & 53 & 66 \\
		0 & 20 & 60 & 65 \\
		0 & 2 & 63 & 64 \\
		0 & 12 & 41 & 59 \\
		0 & 16 & 44 & 58 \\
		0 & 19 & 49 & 57 \\
		0 & 7 & 33 & 55 \\
		0 & 17 & 51 & 54 \\
		0 & 21 & 46 & 52 \\
		0 & 11 & 35 & 50 \\
		0 & 9 & 32 & 36 \\
	\end{array}
	\quad
	\begin{array}{rrrr}
		0 & 14 & 59 & 72 \\
		0 & 16 & 49 & 71 \\
		0 & 5 & 67 & 70 \\
		0 & 21 & 52 & 69 \\
		0 & 24 & 53 & 68 \\
		0 & 9 & 60 & 66 \\
		0 & 23 & 63 & 64 \\
		0 & 11 & 43 & 61 \\
		0 & 20 & 46 & 56 \\
		0 & 7 & 35 & 54 \\
		0 & 4 & 34 & 42 \\
		0 & 12 & 37 & 39 \\
	\end{array}\\[5pt]
	\begin{array}{rrrr}
		0 & 17 & 66 & 78 \\
		0 & 22 & 62 & 77 \\
		0 & 2 & 73 & 76 \\
		0 & 10 & 68 & 75 \\
		0 & 20 & 47 & 72 \\
		0 & 24 & 54 & 70 \\
		0 & 21 & 56 & 69 \\
		0 & 23 & 59 & 67 \\
		0 & 19 & 60 & 64 \\
		0 & 26 & 57 & 63 \\
		0 & 11 & 39 & 53 \\
		0 & 18 & 50 & 51 \\
		0 & 9 & 38 & 43 \\
	\end{array}
	\quad
	\begin{array}{rrrr}
		0 & 6 & 77 & 84 \\
		0 & 18 & 57 & 83 \\
		0 & 8 & 80 & 82 \\
		0 & 23 & 53 & 81 \\
		0 & 12 & 55 & 79 \\
		0 & 3 & 63 & 76 \\
		0 & 27 & 64 & 75 \\
		0 & 14 & 45 & 70 \\
		0 & 22 & 54 & 69 \\
		0 & 19 & 52 & 68 \\
		0 & 20 & 62 & 66 \\
		0 & 17 & 51 & 61 \\
		0 & 21 & 50 & 59 \\
		0 & 5 & 40 & 41 \\
	\end{array}
	\quad
	\begin{array}{rrrr}
		0 & 5 & 76 & 90 \\
		0 & 1 & 82 & 89 \\
		0 & 21 & 58 & 87 \\
		0 & 19 & 59 & 86 \\
		0 & 4 & 64 & 84 \\
		0 & 28 & 72 & 83 \\
		0 & 18 & 63 & 79 \\
		0 & 24 & 70 & 78 \\
		0 & 15 & 51 & 77 \\
		0 & 23 & 73 & 75 \\
		0 & 25 & 57 & 74 \\
		0 & 13 & 47 & 69 \\
		0 & 30 & 65 & 68 \\
		0 & 12 & 43 & 53 \\
		0 & 6 & 39 & 48 \\
	\end{array}
\end{align*}

\subsection{$(1 \leq L \leq 10, L\neq 9, M=4)$}

The cases of $L \in \{6,8\}$ are due to
Laufer \cite{Laufer,Shearer-IBM-Optimal-List,Shearer-IBM-Optimal-List-Credits}, 
the case of $L = 10$
is due to Khansefid et al.\ \cite{KTG,Shearer-IBM-Optimal-List,Shearer-IBM-Optimal-List-Credits}, and 
the cases of 
$L \in \{5,7\}$ were found with Shearer's 
programs \cite{Shearer-IBM-Programs,Shearer-IBM-Programs-Paper}.

\begin{align*}
	\begin{array}{rrrrr}
		0 & 1 & 4 & 9 & 11 \\
	\end{array}
	\quad
	\begin{array}{rrrrr}
		0 & 2 & 9 & 21 & 22 \\
		0 & 4 & 10 & 15 & 18 \\
	\end{array}\\[5pt]
	\begin{array}{rrrrr}
		0 & 2 & 10 & 19 & 32 \\
		0 & 3 & 15 & 26 & 31 \\
		0 & 1 & 7 & 21 & 25 \\
	\end{array}
	\quad
	\begin{array}{rrrrr}
		0 & 4 & 16 & 34 & 41 \\
		0 & 13 & 23 & 32 & 40 \\
		0 & 3 & 24 & 38 & 39 \\
		0 & 5 & 11 & 31 & 33 \\
	\end{array}\\[5pt]
	\begin{array}{rrrrr}
		0 & 5 & 19 & 40 & 42 \\
		0 & 7 & 15 & 33 & 39 \\
		0 & 9 & 22 & 34 & 38 \\
		0 & 1 & 11 & 28 & 31 \\
	\end{array}
	\quad
	\begin{array}{rrrrr}
		0 & 6 & 20 & 48 & 51 \\
		0 & 9 & 21 & 46 & 50 \\
		0 & 13 & 23 & 47 & 49 \\
		0 & 5 & 16 & 35 & 43 \\
		0 & 1 & 18 & 33 & 40 \\
	\end{array}\\[5pt]
	\begin{array}{rrrrr}
		0 & 14 & 26 & 51 & 60 \\
		0 & 4 & 28 & 44 & 59 \\
		0 & 10 & 23 & 52 & 58 \\
		0 & 1 & 21 & 54 & 57 \\
		0 & 7 & 18 & 45 & 50 \\
		0 & 2 & 19 & 41 & 49 \\
	\end{array}
	\quad
	\begin{array}{rrrrr}
		0 & 8 & 28 & 67 & 71 \\
		0 & 10 & 33 & 57 & 70 \\
		0 & 5 & 34 & 55 & 69 \\
		0 & 12 & 27 & 65 & 68 \\
		0 & 1 & 26 & 45 & 62 \\
		0 & 7 & 18 & 49 & 58 \\
		0 & 6 & 22 & 52 & 54 \\
	\end{array}\\[5pt]
	\begin{array}{rrrrr}
		0 & 19 & 34 & 73 & 80 \\
		0 & 8 & 35 & 63 & 79 \\
		0 & 12 & 33 & 74 & 78 \\
		0 & 13 & 30 & 72 & 77 \\
		0 & 11 & 36 & 67 & 76 \\
		0 & 18 & 32 & 69 & 75 \\
		0 & 2 & 22 & 60 & 70 \\
		0 & 1 & 24 & 50 & 53 \\
	\end{array}
	\quad
	\begin{array}{rrrrr}
		0 & 1 & 45 & 98 & 100 \\
		0 & 9 & 36 & 77 & 96 \\
		0 & 14 & 37 & 88 & 95 \\
		0 & 10 & 35 & 83 & 94 \\
		0 & 15 & 46 & 76 & 93 \\
		0 & 12 & 40 & 79 & 92 \\
		0 & 22 & 42 & 85 & 91 \\
		0 & 8 & 34 & 72 & 90 \\
		0 & 3 & 32 & 65 & 89 \\
		0 & 5 & 21 & 71 & 75 \\
	\end{array}
\end{align*}


\end{document}